\documentclass[onefignum,onetabnum]{siamart190516}

\usepackage{lipsum}
\usepackage{amsfonts}
\usepackage{graphicx}
\usepackage{epstopdf}
\usepackage{algorithmic}
\ifpdf
\DeclareGraphicsExtensions{.eps,.pdf,.png,.jpg}
\else
\DeclareGraphicsExtensions{.eps}
\fi

\usepackage{irt-common}
\usepackage{irt-shortcuts}
\usepackage{irt-tikz}

\usepackage{enumitem}

\usepackage{xargs}    
\usepackage{todonotes}
\usepackage{xcolor}
\newcommandx{\unsure}[2][1=]{\todo[inline,linecolor=red,backgroundcolor=red!25,bordercolor=red,#1]{#2}}
\newcommandx{\change}[2][1=]{\todo[inline,linecolor=blue,backgroundcolor=blue!25,bordercolor=blue,#1]{#2}}
\newcommandx{\info}[2][1=]{\todo[inline,linecolor=OliveGreen,backgroundcolor=OliveGreen!25,bordercolor=OliveGreen,#1]{#2}}
\newcommandx{\improve}[2][1=]{\todo[inline,linecolor=orange,backgroundcolor=orange!25,bordercolor=orange,#1]{#2}}

\newsiamremark{remark}{Remark}
\newsiamremark{hypothesis}{Hypothesis}
\newsiamremark{example}{Example}
\crefname{hypothesis}{Hypothesis}{Hypotheses}
\newsiamthm{claim}{Claim}

\headers{Detectability and State Estimation for LTV and Nonlinear Systems}{M. Tranninger, R. Seeber, M. Steinberger, M. Horn, and C. Pötzsche}

\title{Detectability Conditions and State Estimation for Linear Time-Varying and Nonlinear Systems\thanks{This manuscript was accepted for publication in a future issue of the SIAM Journal on Control and Optimization.
	}}

\author{Markus Tranninger\thanks{Institute of Automation and Control, Graz University of Technology, Graz, Austria. 
		(\email{markus.tranninger@tugraz.at}).} %
	\and Richard Seeber \thanks{Christian Doppler Laboratory for Model-Based Control of Complex Test Bed Systems, Institute of Automation and Control, Graz University of Technology, Graz, Austria. }
	\and Martin Steinberger\footnotemark[2]
	\and Martin Horn\footnotemark[3]
	\and Christian Pötzsche \thanks{Department of Mathematics, Alpen-Adria Universität Klagenfurt, Klagenfurt, Austria.}
   }

\usepackage{amsopn}

\def\Chiupper{\lambda^\mathrm{s}}
\def\Chilower{\lambda^\mathrm{i}}
\def\muupper{\mu^\mathrm{s}}
\def\Lupper{\lambda^\mathrm{s}}
\def\Llower{\lambda^\mathrm{i}}
\def\Bupper{\beta^\mathrm{s}}
\def\Blower{\beta^\mathrm{i}}

\def\Bnumlo{\beta^{H,\mathrm{i}}}
\def\Bnumup{\beta^{H,\mathrm{s}}}

\def\SigmaED{\ensuremath{\Sigma_{\mathrm{ED}}}}
\def\SigmaL{\ensuremath{\Sigma_{\mathrm{L}}}}
\def\PL{\ensuremath{\P_{\mathrm{L}}}}

\newcommand{\changed}{}

\begin{document}

\maketitle

\begin{abstract}
 This work proposes a detectability condition for linear time-varying systems based on the exponential dichotomy spectrum. 
 The condition guarantees the existence of an observer, whose gain is determined only by the unstable modes of the system. 
 This allows for an observer design with low computational complexity compared to classical estimation approaches.
 An extension of this observer design to a class of nonlinear systems is proposed and local convergence of the corresponding estimation error dynamics is proven.
 Numerical results show the efficacy of the proposed observer design technique.
\end{abstract}



\section{Introduction}

In recent years, detectability analysis of complex dynamical systems gained a lot of research attention in different engineering fields ~\cite{Bocquet2017,frank2018adetectability,hoger2019ontherelation,tranninger2020detectability,wang2020detectability}.
Despite the fact that detectability is very well studied for linear time invariant systems~\cite{sontag2013mathematical}, the underlying theory and existing detectability conditions are less mature for the time varying and/or nonlinear setting.
In control theory, it is well known that detectability is closely related to the existence of observers, i.e., dynamical systems, which estimate the state of some physical system based on measurements and a system model~\cite{sontag2013mathematical}.
In this work, a system model of the form
\begin{align}\label{eq:NL}
	\dot \x(t) &= \f(\x,\u),\quad \y(t) = \C(t)\x(t), \quad t\in\changed{\mathds R_{\geq 0}} %
\end{align}
is assumed to be known. 
Here, $\x(t)\in \mathds R^n$, $\u(t)\in \mathds R^m$ and $\y(t)\in \mathds R^p$ are the system's state, input and output, respectively.
In order to guarantee uniqueness of a solution to the initial value problem~\eqref{eq:NL} with the initial time $t_0\in\changed{\mathds R_{\geq 0}}$ and the initial state $\x(t_0)=\x_0$, it is assumed that $\f:\mathds R^n \times \mathds R^m\rightarrow \mathds R^n$ is locally Lipschitz continuous in the first argument and piecewise continuous in the second argument.
\changed{The solution is assumed to exist for all $t\geq t_0$.}
Similar to~\cite{frank2018adetectability,KONRAD1998}, the considered observer is assumed to be of the form
\begin{equation}\label{eq:NL:obsv}
	\dot{\hat\x}(t) = \f(\hat\x,\u) + \L(t,\hat\x)\left[\y(t)-\C(t)\hat\x(t)\right], \quad \hat \x(t_0) = \hat \x_0.
\end{equation}
The observer gain $\L:\changed{\mathds R_{\geq 0}} \times \mathds R^n \rightarrow \mathds R^{n\times p}$, which is continuous in the first argument and locally Lipschitz continuous in the second argument, has to be designed.
The goal of this work is to present conditions for the existence of a gain $\L$ that guarantees (at least locally) a uniform exponential convergence of the estimation error $\e(t)=\x(t)-\hat\x(t)$.

State estimation problems do not exclusively occur in control theory, but also, e.g., in synchronization problems for chaotic systems~\cite{pecora1997fundamentals,brown1997synchronization,boccaletti2002thesynchronization,pecora2015chaos}.
Such problems arise for example due to agents interacting in complex networks and are of key importance in biological, chemical or social processes~\cite{arenas2008synchronization}.
In the synchronization terminology, system~\eqref{eq:NL} in combination with~\eqref{eq:NL:obsv} is referred to as a driver-receiver process~\cite{pecora2015chaos} and the challenge is to design a coupling gain $\L$ such that the receiver $\hat\x$ is synchronized with the driver $\x$. 

In geosciences, optimal state estimation is also known as data assimilation~\cite{trevisan2011onthekalman,carrassi2008dataassimilation,Bocquet2017,carrassi2018dataassimilation}.
Data assimilation is a key ingredient for many algorithms to solve problems in, e.g., weather forecasting, oceanography or numerical geology.
There, usually large scale nonlinear models with hundreds to thousands of states are considered together with significant model and measurement uncertainties. 
Hence it is vital for any numerical estimation scheme to confine the estimation and prediction to the part of the system, which dominates the dynamical behavior~\cite{trevisan2011onthekalman,carrassi2008dataassimilation,maclean2021particle}. 
A comprehensive overview of relevant methods and applications in these fields can be found in~\cite{asch2016dataassimilation}.

\subsection{Related Work and Contribution}

A key concept in data assimilation is the assimilation on the unstable subspace (AUS), which was introduced by A. Trevisan and collaborators~\cite{trevisan2011onthekalman} for nonlinear stochastic discrete time systems.
The resulting algorithms are modifications of the extended Kalman filter, where the computation of the error covariances and the Kalman gain is confined to the unstable and neutral tangent subspace of the solutions of the full order filter~\cite{trevisan2011onthekalman,carrassi2008dataassimilation}.
With this approach, the solution of the filter Riccati equation has to be computed only for the unstable and neutral tangent subspace in order to obtain the filter gain.
Usually, this subspace is of significantly lower dimension than the system's state space and hence the computational complexity of the resulting estimation algorithm can be drastically reduced compared to the classical extended Kalman filter (EKF).
Since the initial investigations~\cite{trevisan2011onthekalman}, algorithms for continuous-time~\cite{frank2018adetectability} and discrete-time systems~\cite{carrassi2008dataassimilation,palatella2013Lyapunov} were presented and analyzed.

In~\cite{frank2018adetectability}, a data assimilation scheme based on the numerical approximation of regular Lyapunov exponents for continuous time nonlinear systems is presented. 
The resulting observer design is numerically very efficient, because it does not require solving a differential Riccati equation to compute the observer feedback gain.
The underlying existence condition is, however, conservative.
It requires the number of linearly independent measurements to be at least equal to the number of non-negative Lyapunov exponents. 
Hence, this condition is not fulfilled, for example, by the simple observable linear time invariant double integrator system where the first integrator state is measured.

For linear and nonlinear stochastic discrete time systems, AUS algorithms are discussed in~\cite{trevisan2011onthekalman,carrassi2008dataassimilation,carrassi2018dataassimilation,Bocquet2017,maclean2021particle}.
Such algorithms are very popular, because many numerical studies showed that AUS algorithms are at least as efficient but simpler to implement and computationally less demanding than their original EKF counterparts~\cite{palatella2013Lyapunov}.
Based on the AUS ideas, a detectability condition is introduced in~\cite{grudzien2017asymptotic,Bocquet2017}. 
There, the convergence of the error covariance matrices of the AUS Kalman filter is investigated for linear systems in the presence of bounded additive model errors.
The proposed detectability condition has some similarity to the results presented in~\cite{tranninger2020detectability,tranninger2020uniform} for the deterministic continuous time setting.
Albeit numerical~\cite{tranninger2020efficient,carrassi2008dataassimilation} and theoretical investigations, cf.~\cite{tranninger2020detectability,grudzien2017asymptotic,Bocquet2017}, have shown the validity of the underlying AUS ideas, a formal proof of the estimation error convergence for nonlinear systems is still missing to the best of the authors' knowledge.

This work presents a detectability criterion and an observer design strategy for systems of the form~\eqref{eq:NL}.
It is a significant extension of the ideas presented in~\cite{tranninger2020uniform,tranninger2020detectability} and relies on numerically stable algorithms to compute the dichotomy spectrum of dynamical systems~\cite{Dieci2008}.
The spectral intervals of the dichotomy spectrum are associated with the manifolds comprising solutions with a common exponential growth rate~\cite{siegmund2002dichotomy}.
For linear time invariant systems, the proposed detectability definitions and conditions reduce to well known criteria~\cite{sontag2013mathematical}.
Based on ideas similar to AUS strategies, the observer feedback gain is computed only on a reduced order state space, which corresponds to the unstable and neutral tangent space of the observer trajectory.
Compared to the deterministic extended Kalman-Bucy filter, this results in a numerically efficient observer for systems with a large system order.
\changed{The existence conditions for the observer are shown to be fulfilled for a class of uniformly observable nonlinear systems.} 
Moreover, the present work is a step towards closing the gap between theory and applications of AUS algorithms, because it provides a convergence proof of the resulting observer in a deterministic continuous time framework.

The paper is structured as follows: \Cref{sec:prel} introduces required tools for the stability analysis of linear time varying and nonlinear systems.
\Cref{sec:spectra} summarizes the Lyapunov and exponential dichotomy spectrum, which are crucial ideas for the detectability condition and observer design presented in~\Cref{sec:EDObserverdesign} for linear time varying systems.
In~\Cref{sec:NLobsv}, the results obtained in~\Cref{sec:EDObserverdesign} are then extended to nonlinear systems of the form~\eqref{eq:NL} following the ideas of the (deterministic) extended Kalman-Bucy filter.
Detailed numerical investigations of the proposed observer algorithm are presented in~\Cref{sec:example} for the chaotic nonlinear Lorenz'96 model~\cite{lorenz96predictability}.
Together with the conclusion, further research directions are pointed out in~\Cref{sec:conclusions}.
It should be remarked that the results presented in~\Cref{sec:EDObserverdesign} extend the authors' original research results~\cite{tranninger2020detectability,tranninger2020uniform}.
Material already published in~\cite{tranninger2021beobachterentwurf} is partially included in~\Cref{sec:EDObserverdesign} in order to make the paper self-contained.

\emph{Notation:}  Matrices are printed in boldface capital letters, whereas column vectors are boldface lower case letters. 
The matrix $\I_n$ is the $n\times n$ identity matrix.
\changed{Moreover, $\M = \diag{\M_1,\ldots,\M_j}$ denotes a (block) diagonal matrix with entries $\M_1$, $\ldots$, $\M_j$.}
The 2-norm of a vector or the corresponding induced matrix norm is denoted by $\|\cdot\|$.
Symmetric positive definite (positive semidefinite) matrices $\M^\transp = \M$ are denoted by $\M \succ 0$ ($\M \succeq 0$). 
If for two symmetric matrices $\M_1$, \changed{$\M_2$} it holds that $\M_2 - \M_1 \succ 0$ ($\succeq 0$), then $\M_1\prec \M_2$ ($\M_1\preceq \M_2$).
In dynamical systems, differentiation of a vector $\x$ with respect to time $t$ is expressed as $\dot \x$.
When writing such systems, time dependence of state (usually $\x$) and output (usually $\y$) is suppressed and only time dependence of the system's parameters is stated explicitly. 
\changed{Moreover, $\mathds R$ and $\mathds{R}_{\geq 0}$ denote the reals and non-negative reals, respectively.}

\section{Preliminaries}\label{sec:prel}

This section briefly summarizes stability results for dynamical systems.
\changed{
	First, stability properties of solutions of differential equations
	\begin{equation}\label{eq:system_nonlinear_Lyap}
		\dot \x = \f(t,\x)
	\end{equation}
	with a continuous function $\f:\mathds{R}_{\geq 0} \times \mathds R^n \mapsto \mathds R^n$ are considered 
	for times $t\in \changed{\mathds R_{\geq 0}}$.
	The initial state at the initial time $t_0\in\changed{\mathds R_{\geq 0}}$ is denoted by $\x(t_0)=\x_0$.
	It is assumed that the solution of~\eqref{eq:system_nonlinear_Lyap} exists for all $t\geq t_0$.}

For the stability assessment, one is often interested in the behavior of solutions starting near an equilibrium $\x_r$ fulfilling $\vc 0 = \f(t,\x_r)$.
It is assumed that $\x_r = \vc 0$ is an equilibrium\footnote{For \changed{assessing stability of} any other solution $\x_r(t)$, consider new coordinates $\z(t) = \x(t)-\x_r(t)$.} of~\eqref{eq:system_nonlinear_Lyap} in the following.

Details on different stability concepts can be found, e.g., in~\cite{Hahn1967,Hinrichsen2006}.
The notions used in this paper are introduced in the following %
\begin{definition}[stability notions]\label{def:stability}
	The \changed{equilibrium} $\x_r=\vc 0$ of~\eqref{eq:system_nonlinear_Lyap}, is called
	\begin{enumerate}[label=(\roman*)]
		\item \changed{(locally)} exponentially stable, if for some real $\mu>0$ and every $t_0\in\mathds J$ there exist scalars $\rho(t_0)>0$ and $K(t_0)\geq 1$ such that for every $\x_0$ with $\|\x_0\|\leq \rho(t_0)$ one has
		\begin{equation}\label{eq:exponential_stability}
			\|\x(t)\| \leq K(t_0)e^{-\mu(t-t_0)}\|\x(t_0)\|\;\;\text{ for all $t\geq t_0$;}
		\end{equation} \label{def:stability:ES}
		\item \changed{(locally)} uniformly exponentially stable, if $K$, $\rho$ in (ii) are independent of $t_0$;\label{def:stability:UES}
		\item globally exponentially stable or globally uniformly exponentially stable, if~\ref{def:stability:ES} or~\ref{def:stability:UES} is fulfilled for all $\x_0\in\mathds R^n$, respectively.\label{def:stability:GES}
	\end{enumerate}
\end{definition}

\subsection{Stability of Linear Time-Varying Systems}\label{sec:prel_stability}
In the following, the linear time-varying autonomous system
\begin{equation}\label{eq:sys:LTV}
	\dot \x = \A(t)\x
\end{equation}
with $\x(t) \in\mathds{R}^n$ is considered for $t \in \changed{\mathds R_{\geq 0}}$.
It is assumed that $\A: \changed{\mathds R_{\geq 0}}\rightarrow \mathds R^{n\times n}$ is continuous and bounded.
The following statements can be found in standard textbooks, see, e.g.,~\cite{Hinrichsen2006}.
System~\eqref{eq:sys:LTV} has the unique solution 
$	\x(t) = \mtPhi(t,t_0)\x_0,$
where $\mtPhi:\changed{\mathds R_{\geq 0}}\times \changed{\mathds R_{\geq 0}}\rightarrow \mathds R^{n\times n}$ is the state transition matrix, $\x(t_0)=\x_0\in\mathds{R}^n$ is the initial state and $t_0\in\changed{\mathds R_{\geq 0}}$ is the considered initial time. 
The state transition matrix can be obtained from the associated fundamental matrix differential equation
\begin{equation}\label{eq:fundamental_prel}
	\dot{\X} = \A(t) \X, \quad \X(t)\in\mathds{R}^{n\times n}.
\end{equation}
For any solution of~\eqref{eq:fundamental_prel} with $\X(0)=\X_0$ as a non-singular matrix, the state transition matrix is given by
$\mtPhi(t,t_0) = \X(t)\X^{-1}(t_0).$

For a linear system~\eqref{eq:sys:LTV}, \changed{ all equilibria possess identical stability properties, which are entirely characterized by the} state transition matrix $\mtPhi$ \changed{ according to:}
\begin{lemma}[stability criteria for linear systems]\label{thm:stabilitycriteria}
	System~\eqref{eq:sys:LTV} is
	\begin{enumerate}[label=(\roman*)]
		\item globally exponentially stable, if and only if there exists a constant $\mu>0$ such that for every $t_0\in \changed{\mathds R_{\geq 0}}$ there exists a scalar $K(t_0)\geq 1$ such that
		\begin{equation}\label{eq:stab:es}
			\|\mtPhi(t,t_0)\| \leq K(t_0)e^{-\mu(t-t_0)} \text{ for all }t\geq t_0%
		\end{equation}
		\item globally uniformly exponentially stable, if and only if $K$ in (ii) is independent of $t_0$;
	\end{enumerate}
\end{lemma}
These relations are well known in the literature, see, e.g.~\cite[Chapter VIII]{Hahn1967}.

A state transformation, which preserves the stability properties of the system is called a Lyapunov transformation.
It is a smooth and invertible linear change of coordinates $\z(t)=\T(t)\x(t)$, where $\T(t)$, $\T^{-1}(t)$ and $\dot{\T}(t)$ are uniformly bounded for all $t\in\changed{\mathds R_{\geq 0}}$, see~\cite{Adrianova1995}.
The following results presented in~\cite[Chapter 3.3.4]{Hinrichsen2006} allow to construct a Lyapunov function for uniformly exponentially stable systems.
\begin{proposition}[Lyapunov function for linear time-varying systems]\label{thm:LyapLTV}
	Supposed that~\eqref{eq:sys:LTV} is uniformly exponentially stable, then:
	\begin{enumerate}[label=(\roman*)] 
		\item\label{item:uniqueness} There exists a unique bounded solution $\P:\changed{\mathds R_{\geq 0}} \rightarrow \mathds R^{n\times n}$ of the matrix differential equation
		\begin{equation}\label{eq:Lyapeq}
			\dot \P + \A^\transp(t)\P + \P\A(t) + \Q(t) = \vc 0
		\end{equation}
		for any $\Q(t)$ with $q_1\I_n\preceq \Q(t) \preceq q_2\I_n$ and any positive constants $q_1$, $q_2$.
		\item\label{item:diffP} The only bounded $\P$ which solves~\eqref{eq:Lyapeq} is given by
		\begin{equation}
			\P(t) =\integ{s}{t}{\infty}{\mtPhi^\transp(s,t)\Q(t)\mtPhi(s,t)} \text{ with } t\in\changed{\mathds R_{\geq 0}},
		\end{equation}
		and there exist positive constants $p_1,\,p_2$ such that $p_1\I_n \preceq \P(t) \preceq p_2\I_n$.
		\item\label{item:Lyapfcn} $V:\changed{\mathds R_{\geq 0}}\times \mathds R^n \rightarrow \changed{\mathds R_{\geq 0}}$ defined as $V(t,\x) = \x^\transp \P(t) \x$ is a Lyapunov function for~\eqref{eq:sys:LTV}, whose time derivative $\dot V = \frac{\partial V}{\partial t} + \frac{\partial V}{\partial x}\A(t)\x$ along the trajectory of~\eqref{eq:sys:LTV} satisfies
		\begin{equation}
			\dot V(t,\x) = -\x^\transp\Q(t)\x.
		\end{equation}
		\item \label{item:p2bound} The constant $p_2$ can be bounded according to
		$		p_2 \leq \frac{K^2 q_2}{2\mu}$,
		where $K$ and $\mu$ are obtained from the bound on the state transition matrix~\eqref{eq:stab:es}.
	\end{enumerate}
\end{proposition}
\changed{Items~\ref{item:uniqueness} and~\ref{item:diffP} follow from~\cite[Lemma 3.3.36 and Theorem 3.3.38]{Hinrichsen2006}. 
	Item~\ref{item:Lyapfcn} follows from~\cite[Theorem 3.3.33 and eqs. (55)--(56)]{Hinrichsen2006} and item~\ref{item:p2bound} is given in \cite[eq. (62)]{Hinrichsen2006}.}
For time-invariant systems, the uniform and non-uniform stability notions coincide. Moreover, an eigenvalue spectrum $\sigma(\A)$ with negative real parts is equivalent to uniform exponential stability.
In general, the (time-dependent) eigenvalues of a time-varying coefficient matrix do not allow to reason about the stability properties~\changed{\cite[p. 257]{Hinrichsen2006}}.
In the following, two concepts for the generalization of eigenvalues to the time-varying setting will be discussed together with the numerical computation. 

\section{Two Important Spectra for Linear Time-Varying Systems}\label{sec:spectra}
The Lyapunov spectrum and the exponential dichotomy spectrum are two important generalizations of the eigenvalue spectrum to the time-varying case. 
A negative Lyapunov spectrum guarantees exponential stability~\changed{\cite[p. 6]{barreira2002lyapunov}}, whereas a negative exponential dichotomy spectrum guarantees uniform exponential stability.
Both spectra are introduced in the following based on~\cite{Dieci2007}.

\subsection{Lyapunov-Spectrum}
Let the real functional
\begin{align}\label{eq:definition_LEfunction}
	\Chiupper(\x) &= \limsup_{t\rightarrow \infty} \frac{1}{t}\ln \|\x(t) \| 
\end{align}
be defined for non-zero functions $\x:\changed{\mathds R_{\geq 0}} \rightarrow \mathds R^n$.
This functional provides an asymptotic upper bound of the exponential growth or decay for a given non-trivial solution $\x$ of~\eqref{eq:sys:LTV}. 
Note that this bound does not depend on $\x$ on a finite time interval and hence it can be assumed that $t_0=0$.

The functional~\eqref{eq:definition_LEfunction} is now applied to the solutions $\x_i(t)=\X(t)\e_i$ yielding $\changed{\ell}_i=\Chiupper(\x_i)$, where $\X:\changed{\mathds R_{\geq 0}} \rightarrow \mathds R^{n\times n}$ is an arbitrary fundamental solution satisfying~\eqref{eq:fundamental_prel} and $\e_i$ the $i$-th standard basis vector.
Minimizing $\sum_{i=1}^n\changed{\ell}_i$ over all possible fundamental solutions $\X$ results in the so-called \emph{upper} Lyapunov exponents $\Chiupper_i\coloneqq \changed{\ell}_i$ of system~\eqref{eq:sys:LTV}, see~\cite{Dieci2007}.
Without loss of generality, the Lyapunov-Exponents are assumed to be ordered\footnote{The ordering can always be achieved by a column permutation of $\X(t)$.} according to $\Chiupper_1\geq \Chiupper_2 \geq \ldots \geq\Chiupper_n$. 
The corresponding fundamental solution is then called an \emph{ordered normal Lyapunov basis} and the $n$ Lyapunov exponents are unique.
\changed{In the case of identical Lyapunov exponents, 
	their multiplicity is determined by the dimension of the solution space with the same exponent.}

The growth rate of all trajectories of~\eqref{eq:sys:LTV} can be bounded based on the largest upper Lyapunov-Exponent $\Chiupper_1$.
More specifically, for all $\epsilon>0$, there exists a $K_\epsilon\geq 1$ such that
$	\|\mtPhi(t,0)\| \leq K_\epsilon e^{(\Chiupper_1+\epsilon)t}$, \changed{see \cite[p. 6]{barreira2002lyapunov}.}
Hence, $\Chiupper_1<0$ implies exponential stability, because for a fixed $t_0$ and a sufficiently small $\epsilon>0$, the scalars in~\eqref{eq:stab:es} can be chosen according to $\mu=-(\Chiupper_1+\epsilon)>0$ and \mbox{$K(t_0) = K_\epsilon e^{(\Chiupper_1+\epsilon)t_0}\|\mtPhi(0,t_0)\|$}.

In general, there also exists an asymptotic lower bound on the growth rate of a trajectory, which does not necessarily coincide with the corresponding upper Lyapunov exponent, see, e.g. the example given in~\cite{perron1930dieordnungszahlen}.
These lower bounds can be obtained by considering the adjoint system 
\begin{equation}\label{eq:adjSystem}
	\dot{\bm\chi} = -\A^\transp(t)\bm\chi.
\end{equation}
Denoting the sorted upper Lyapunov exponents of~\eqref{eq:adjSystem} by $-\muupper_i$ with $-\muupper_1\leq -\muupper_2 \leq \cdots \leq -\muupper_n$ 
allows to introduce the \emph{lower} Lyapunov exponents
$	\Chilower_i = \muupper_i,\quad i=1,\ldots,n$
for system~\eqref{eq:sys:LTV}.
In general, it holds that $\Chilower_i\leq \Chiupper_i$~\cite{Dieci2007,lyapunov1992thegeneral}.
In~\cite{Dieci2003}, a spectrum based on the lower and upper Lyapunov exponents is proposed according to the following %
\begin{definition}[Lyapunov spectrum]\label{def:LyapSpektrum}
	The Lyapunov spectrum $\Sigma_L$ of~\eqref{eq:sys:LTV} is
	$\Sigma_L = \bigcup_{i=1}^n [\Chilower_i,\Chiupper_i].$
\end{definition}

For periodic systems and, in particular, time-invariant systems, the Lyapunov spectrum reduces to isolated points and it holds that $\Chilower_i=\Chiupper_i$.
Moreover, the Lyapunov exponents coincide with the real parts of the Floquet exponents for periodic systems or the eigenvalues for time-invariant systems, see~\cite[Theorem 63.4]{Hahn1967}.

The larger class of regular systems involves all linear time-varying systems, where the Lyapunov spectrum is a set of isolated points. 
The concept of regularity was introduced by Lyapunov~\cite{lyapunov1992thegeneral} and is often demanded for the numerical computation of the Lyapunov exponents.
\begin{definition}[regularity,~{{\cite[Def. 64.1]{Hahn1967}}}]\label{def:regularity}
	System~\eqref{eq:sys:LTV} is called regular, if $\Chilower_i=\Chiupper_i$ for all $i=1,2,\ldots,n$. 
	\changed{In this case, set $\lambda_i=\Chilower_i=\Chiupper_i$.}
\end{definition}

Regularity is hard to verify in practice for a specific system  \changed{which motivates to employ the exponential dichotomy spectrum as introduced in the following.}

\subsection{Exponential Dichotomy Spectrum} 
An important idea for the stability analysis of linear time-varying systems is the exponential dichotomy introduced by O. Perron as a generalization of hyperbolicity to the time-varying case, see~\cite{siegmund2002dichotomy}, \cite[Ch. 4, \parsymb 3]{daleckii2002stability} or \cite{coppel1978dichotomies}.

\begin{definition}[exponential dichotomy]\label{def:ED}
	\changed{A system $\dot \x = \A(t)\x$} possesses an exponential dichotomy, if there exists a projection matrix\footnote{A projection matrix is a matrix satisfying $\P^2=\P$.} $\P\in\mathds{R}^{n\times n}$, a fundamental solution $\X$ and real constants $K\geq 1$ and $\alpha>0$ such that 
	\begin{subequations}\label{eq:EDdef}
		\begin{align}
			&	\|\X(t)\P\X^{-1}(t_0)\| \leq Ke^{-\alpha(t-t_0)} \text{ for } t\geq t_0\geq 0 \text{ and } \label{eq:EDstable}\\
			&	\|\X(t)(\I_n-\P)\X^{-1}(t_0)\|  \leq Ke^{\alpha (t-t_0)} \text{ for } 0\leq t\leq t_0.\label{eq:EDunstable}
		\end{align} 
	\end{subequations}	
\end{definition}
It follows directly, that system~\eqref{eq:sys:LTV} is uniformly exponentially stable, if and only if it has an exponential dichotomy with $\P=\I_n$. 
\changed{In the time varying case, the system has an exponential dichotomy if and only if it has no eigenvalues on the imaginary axis.}
A spectrum based on the exponential dichotomy is introduced in~\cite{Sacker1978,siegmund2002dichotomy} as follows.
\begin{definition}[dichotomy spectrum,~{{{\cite[Definition 3.1]{siegmund2002dichotomy}}}}]\label{def:dichotomyspectrum}
	The dichotomy spectrum \SigmaED\ of~\eqref{eq:sys:LTV} is the set of all $\mu\in\mathds R$ for which the systems
	$		\dot{\bm\xi} = \left[\A(t)-\mu \I_n \right] \bm\xi$
	do \emph{not} have an exponential dichotomy.
\end{definition}
If~\eqref{eq:sys:LTV} is time-invariant, then \changed{$\SigmaED$ is equal to the real part of the eigenvalue spectrum $\sigma(\A)$}.
For systems~\eqref{eq:sys:LTV} with a bounded $\A(t)$, the dichotomy spectrum consists of $1\leq d\leq n$ compact and disjoint subintervals~\cite{siegmund2002dichotomy}.

\begin{remark}
	For a one-dimensional system $\dot x = a(t)x$, \changed{the so-called lower and upper Bohl exponents are defined according to} 
	\begin{align}
		\Blower_1 =  \liminf_{\substack{
				t-t_0\rightarrow \infty\\t_0\rightarrow\infty}}\frac{1}{t}\integ{\tau}{t_0}{t_0+t}{a(\tau)} \text{ and } \;
		\Bupper_1 = \limsup_{\substack{
				t-t_0\rightarrow \infty\\t_0\rightarrow\infty}}\frac{1}{t}\integ{\tau}{t_0}{t_0+t}{a(\tau)},
	\end{align}
	\changed{respectively.}
	The one-dimensional system then has an exponential dichotomy, if and only if $0<\Blower_1 \leq \Bupper_1$ or $\Blower_1\leq \Bupper_1 < 0$, see~\cite{daleckii2002stability}. 
	Its dichotomy spectrum is hence given by $\SigmaED=[\Blower_1, \Bupper_1]$.
	Moreover, it holds that $\Blower_1\leq \Llower_1 \leq \Chiupper_1\leq \Bupper_1$ and therefore $\SigmaL\subseteq \SigmaED$.
\end{remark}

\subsection{Numerical Approximation} \label{sec:numerical}
For simplicity, it is assumed that system \eqref{eq:sys:LTV} is regular.
The non-regular case is treated in~\cite{Dieci2007}.
In a first step, system~\eqref{eq:sys:LTV} is transformed to an upper triangular form using a Lyapunov transformation.
For systems with a real coefficient matrix $\A(t)$,~\cite[Theorem 3.3.1 and Remark 3.3.2]{Adrianova1995} state that this change of coordinates $\R(t)=\Q^\transp(t)\X(t)$ can be achieved using an orthogonal Lyapunov transformation matrix $\Q:\changed{\mathds R_{\geq 0}}\rightarrow \mathds R^{n\times n}$.
This transformation can be obtained numerically by means of a continuous QR decomposition~\cite{Dieci2007}.
In this case, $\Q$ and $\R$ are the solutions of 
\begin{subequations}\label{eq:QRfull}
	\begin{align}
		\dot{\R} &= \B(t)\R, \; \quad \B(t) = \Q^\transp(t) \A(t)\Q(t)-\S(t)\label{eq:fullqr1}\\
		\dot{\Q} &= \Q\S(t).\;\label{eq:fullqr2}
	\end{align}
\end{subequations}
The coefficients $s_{ij}(t)$ of the skew symmetric  $n\times n$ matrix $\S(t)$ are given by $s_{ij}(t) = \q_i^\transp(t) \A(t)  \q_j(t)$ for $i>j$, where $\q_i(t)$ denotes the $i$-th column of $\Q(t)$.

In a first step, it is assumed that the initial conditions are obtained by a QR decomposition of an ordered normal Lyapunov basis\footnote{The numerical implementation is discussed in Section~\ref{sec:reducedQR}.}, i.e., $\X_0=\Q_0\R_0$.
Under the assumption that all diagonal elements of $\R_0$ are non-negative, this decomposition is unique.
\changed{For regular systems~\eqref{eq:sys:LTV} transformed to an upper triangular form
	$\dot \z = \B(t)\z$ %
	with $\z(t)=\Q^\transp(t) \x(t)$, the Lyapunov exponents can be obtained by time-averaging of the diagonal elements according to
	\begin{equation}\label{eq:lyapforwardregularity}
		\lambda_i=  \lim_{t\rightarrow \infty} \frac{1}{t}\integ{\tau}{0}{t}{b_{ii}(\tau)}= \lim_{t\rightarrow \infty}\frac{1}{t} \integ{\tau}{0}{t}{\q_i^\transp(\tau)\A({\tau})\q_i(\tau)},
	\end{equation}
	see~\cite{lyapunov1992thegeneral}}.
To obtain the Lyapunov exponents, merely solving the orthogonal differential equation~\eqref{eq:fullqr2} and not the (possibly unbounded) solution of~\eqref{eq:fullqr1} is required.

Independent of the system's regularity, the exponential dichotomy spectrum can always be obtained from the diagonal elements of $\B$, see~\cite{Doan2016}.
To that end, let the two real functionals
\begin{subequations}\label{eq:bohlexponents}
	\begin{equation}
		\Blower(b) = \liminf_{\substack{
				t-t_0\rightarrow \infty\\t_0\rightarrow\infty}} \frac{1}{t}\integ{\tau}{t_0}{t_0+t}{b(\tau)}\quad \text{ and } \quad
		\Bupper(b) = \limsup_{\substack{
				t-t_0\rightarrow \infty\\t_0\rightarrow\infty}} \frac{1}{t}\integ{\tau}{t_0}{t_0+t}{b(\tau)}
	\end{equation}
\end{subequations}
be introduced for a scalar function $b:\changed{\mathds R_{\geq 0}}\rightarrow \mathds R$.
Applied to the diagonal entries of $\B$, these functionals provide the lower and upper Bohl exponents  $\Blower_i=\Blower(b_{ii})$ and $\Bupper_i=\Bupper(b_{ii})$, respectively, of the scalar systems
$\dot \eta_i  = b_{ii}(t)\eta_i$ for $i=1,2,\ldots,n.$
According to~\cite[Proposition 5]{Doan2016}, the dichotomy spectrum is then obtained as
\begin{equation}
	\SigmaED = \bigcup_{i=1}^{n} \Lambda_i \quad \text{with}\quad \Lambda_i = \left[\Blower_i,\Bupper_i\right].
\end{equation}
Note that the $n$ intervals $\Lambda_i$ are not necessarily disjoint and the obtained Bohl exponents are not necessarily sorted.
In many cases, however, a sorting similar to the Lyapunov exponents can be expected.
From the relations in~\eqref{eq:lyapforwardregularity} and \eqref{eq:bohlexponents} it follows that $\Blower_i\leq \lambda_i \leq \Bupper_i$ and $\SigmaL\subseteq\SigmaED$, i.e., the Lyapunov spectrum is contained in the exponential dichotomy spectrum.
This is also true in the non-regular case~\cite{Dieci2007}.
For periodic and time-invariant systems, both spectra coincide, i.e., $\SigmaL = \SigmaED$.
In the general time-varying case, however, this is not true as demonstrated in~\cite{Doan2016}.

The numerical approximation of the Bohl exponents turns out to be more difficult than the approximation of (regular) Lyapunov exponents, because $t_0$ and $t-t_0$ tend to infinity in~\eqref{eq:bohlexponents}.
For the approximation of the Bohl exponents, let
\begin{equation}\label{eq:Bohlapprox}
	\Bnumlo_i = \inf_{t_0\in\changed{\mathds R_{\geq 0}}}\frac{1}{H}\integ{\tau}{t_0}{t_0+H}{b_{ii}(\tau)} \quad\text{ and } \quad
	\Bnumup_i = \sup_{t_0\in\changed{\mathds R_{\geq 0}}}\frac{1}{H}\integ{\tau}{t_0}{t_0+H}{b_{ii}(\tau)}
\end{equation}
be introduced with a scalar parameter $H>0$.
\changed{It is stated in~\cite[Theorem 8.4]{Dieci2003} that for any $H > 0$, $\Lambda_i = [\Blower_i,\Bupper_i] \subseteq [\Bnumlo_i,\Bnumup_i]$.
	It is furthermore claimed that for $H > 0$ sufficiently large it holds that 
	$[\Bnumlo_i,\Bnumup_i] \subseteq [\Blower_i,\Bupper_i]$ and hence $\Lambda_i=[\Bnumlo_i,\Bnumup_i]$.
	The latter statement is not true as shown by the following counterexample. 
	\begin{example}
		Consider the scalar system $\dot x = \frac{1}{1+t}x$ which admits the fundamental solution $X(t)=(1+t)$.
		The exponential dichotomy spectrum is $\Sigma_{\mathrm{ED}} = \{0\}$ and hence $\Blower_1=\Bupper_1 = 0$.
		The computation of $\Bnumup$ gives 
		\begin{equation}
			\Bnumup_1 = \sup_t \frac{1}{H}\integ{\tau}{t}{t+H}{\frac{1}{1+\tau}} = \sup_t \frac{1}{H} \ln(1 + \frac{H}{1+t}) = \frac{1}{H}\ln(1+H)
		\end{equation}
		with a supremum at $t=0$. 
		Hence, for any finite $H>0$ it holds that $\Bnumup_1>0=\Bupper_1$.
	\end{example}
}

The \changed{correct} relation between the \changed{approximation}~\eqref{eq:Bohlapprox} and the bounds of the spectral intervals~\eqref{eq:bohlexponents} is summarized in the following
\begin{theorem}[approximation of the spectral intervals]\label{thm:approximation}
	For every $H>0$, it holds that
	\begin{equation}\label{eq:bapprox1}
		\Bnumlo_i \leq \Blower_i \leq \Bupper_i \leq \Bnumup_i.
	\end{equation}
	Moreover, for every $\epsilon>0$, there exists a $H_0>0$ such that for all $H\geq H_0$ it holds that
	\begin{equation}\label{eq:bapprox2}
		\Blower_i-\epsilon \leq \Bnumlo_i \text{ and } 	\Bnumup_i \leq \Bupper_i+\epsilon.
	\end{equation}
\end{theorem}
The proof is given in the appendix.
Averaging the $b_{ii}$, $i=1,\ldots,n$, over a sufficiently large time interval $H$ allows to approximate the spectral intervals theoretically with arbitrary accuracy.
In particular, Theorem~\ref{thm:approximation} implies that in the limit, it holds that
{$\Bupper_i = \lim_{H\rightarrow \infty} \sup_{t_0} \frac{1}{H}\integ{\tau}{t_0}{t_0+H}{b_{ii}(\tau)}.$
	Detailed considerations and numerical examples for the approximation of the Lyapunov exponents (also for non-regular systems) and the exponential dichotomy spectral intervals can be found in~\cite{Dieci2007}.
	
	\subsection{Reduced QR Decomposition}\label{sec:reducedQR}
	If not all spectral intervals are of interest, it may suffice to compute only a part of the spectrum.
	This allows a reduction of the computational complexity by utilizing properties of the QR-decomposition~\cite{tranninger2020detectability}.
	For the reduced QR decomposition, only the first $k\leq n$ columns of an ordered normal Lyapunov basis $\X$ are considered.
	This results in the differential equations for the reduced QR decomposition
	\begin{subequations}
		\begin{align}
			\dot{\R}_1 &= \B_{1}(t)\R_1, \quad \R_1(0)\in\mathds{R}^{k\times k}, \label{eq:redQR1}\\
			\dot{\bar \Q} &=(\I_n-\bar \Q \bar \Q^\transp)\A(t)\bar \Q +\bar \Q \S_1(t), \quad {\bar \Q(0) \in\mathds{R}^{n\times k},\quad\text{with}}\label{eq:redQR2}   \\
			\B_{1}(t) &=\bar \Q^\transp(t) \A(t) \bar \Q(t)- \S_1(t)\label{eq:b1}.
		\end{align}
		The elements $s_{ij}$ of the skew-symmetric $k\times k$-Matrix $\S_1$ are given by
		$s_{ij}(t) = \bar \q_i^\transp(t) \A(t) \bar \q_j(t)$
		for $i>j$~\cite{frank2018adetectability,tranninger2020detectability}.
	\end{subequations}
	Again, only the solution of~\eqref{eq:redQR2} is required for the approximation of $k$ spectral intervals via the diagonal entries of $\B_{1}$.
	For this, it is crucial to preserve orthogonality of the columns of $\bar \Q$.
	This can be achieved by a projected integration algorithm, which is a standard integration scheme (e.g., a fourth-order Runge-Kutta algorithm) in combination with an orthogonalization procedure based, e.g., on the modified Gram-Schmidt algorithm \changed{as} presented in~\cite{Dieci1994}.
	
	If the initial condition \mbox{$\bar \X(0) = \bar \Q(0) \R_{1}(0)$} is chosen as the first $k$ columns of an ordered normal Lyapunov basis, the $k$ largest Lyapunov exponents can be obtained from the diagonal of $\B_{1}$.
	Typically, such a basis is not known a priory and determining such a basis may be cumbersome~\cite{Dieci2008,Dieci2007}.
	Hence, it is proposed in~\cite{Dieci2008} to choose $\bar\Q(0)$ as random orthogonal matrix in practice.
	This choice proved to be successful in many numerical investigations~\cite{frank2018adetectability,tranninger2020detectability,tranninger2020efficient} and is theoretically supported by~\cite{benettin1980Lyapunov}.

\section{Uniform Detectability of Linear Time-Varying Systems}\label{sec:EDObserverdesign}
This section deals with the existence of an observer for systems of the form~\eqref{eq:sys:LTV}, i.e.,
\begin{align}
	\dot\x &= \A(t)\x,\quad \y = \C(t)\x \label{eq:sys:ltvwithoutput}
\end{align}
with the output $\y(t)\in\mathds{R}^p $. 
The output matrix function $\C:\changed{\mathds R_{\geq 0}} \rightarrow \mathds R^{p\times n}$ is assumed to be continuous and bounded.
The goal is to design an observer
\begin{equation}\label{eq:observer}
	\dot{\hat\x} = \A(t)\hat\x + \L(t)\left[\y-\C(t)\hat\x\right],
\end{equation}
where $\L:\changed{\mathds R_{\geq 0}} \rightarrow \mathds R^{n\times p}$ is a design parameter.
The estimation error $\e(t)=\x(t)-\hat\x(t)$ is governed by
\begin{equation}\label{eq:errordynamics}
	\dot \e = \left[\A(t)-\L(t)\C(t)\right]\e.
\end{equation}
The following detectability definition \changed{introduced} in~\cite{wonham1968onamatrix,ravi1992normalized}, is strongly related to the stability properties of~\eqref{eq:errordynamics}.
\begin{definition}\label{def:detectability}
	System~\eqref{eq:sys:ltvwithoutput} is called uniformly exponentially detectable, if there exists a bounded $\L:\changed{\mathds R_{\geq 0}} \rightarrow\mathds R^{n\times p}$ such that the estimation error dynamics~\eqref{eq:errordynamics} is uniformly exponentially stable.
\end{definition}
In other words, uniform exponential detectability or short \emph{detectability} is equivalent to the existence of an observer with a uniformly exponentially stable estimation error dynamics. 
A stronger concept which also guarantees the existence of an observer is uniform complete observability.
It is characterized by the constructibility Gramian
\begin{equation}
	\N(t_1,t_0) = \integ{s}{t_0}{t_1}{\mtPhi^\transp(s,t_1)\C^\transp(s)\C(s)\mtPhi(s,t_1)}
\end{equation}
according to
\begin{definition}[uniform complete observability]
	System \eqref{eq:sys:ltvwithoutput}, or equivalently the pair $(\A(t),\C(t))$, is called uniformly completely observable, if there exist positive constants $\alpha_1$, $\alpha_2$ and $\changed{T}$, such that
	$	\alpha_1\I_n \preceq \N(t_0+\changed{T},t_0) \preceq \alpha_2\I_n$
	holds for all $t_0\in\changed{\mathds R_{\geq 0}}$.
\end{definition}
It is shown in~\cite{bucy1972thericcati} that uniform complete observability implies the existence of a unique and uniformly bounded positive definite solution of the observer Riccati equation
\begin{equation}\label{eq:fullriccati}
	\dot \P = \A(t)\P + \P \A^\transp(t) - \P \C^\transp(t) \C(t) \P + \G(t).
\end{equation}
Here, $\P(0)\succ 0$ and $\G(t)=\G^\transp(t)$ with $ g_1\I_n\preceq \G(t) \preceq g_2\I_n$ and $0<g_1\leq g_2$. 
Moreover, choosing the feedback gain in~\eqref{eq:errordynamics} as $\L(t)=\P(t)\C^\transp(t)$ renders the error dynamics uniformly exponentially stable~\cite{bucy1967global} and hence, uniform complete observability is sufficient for uniform exponential detectability.

Uniform complete observability is a strong system requirement and solving the differential Riccati equation~\eqref{eq:fullriccati} might be computationally demanding, especially for systems with a large system order.
If the system, however, possesses only a few unstable modes, i.e., only a small number of non-negative upper Bohl exponents, it is reasonable to modify only these modes in the estimation error dynamics.
This allows to solve the Riccati equation on a reduced order subspace.
This strategy is also pursued in~\cite{palatella2013Lyapunov,tranninger2020detectability} for non-negative (regular) Lyapunov exponents and is the underlying idea of the following detectability condition and observer design technique.

\begin{theorem}[uniform detectability]\label{thm:detectability:ltv}
	Consider the linear time-varying system~\eqref{eq:sys:ltvwithoutput}.
	Let $\Q:\changed{\mathds R_{\geq 0}} \rightarrow \mathds R^{n\times n}$ be the solution to~\eqref{eq:fullqr2} of the continuous QR-decomposition of a corresponding fundamental solution $\X$ with $\X(t_0)=\Q(t_0)\R(t_0)$.
	Let $j^\star$ be the first integer $0\leq j^\star \leq n$ such that $\Bupper(\q_i(t)^\transp \A(t)\q_i(t))< 0$ for $i=j^\star+1,\ldots,n$ and let $\Q(t)$ be partitioned according to $\Q(t) = \left[\bar \Q(t) \;\; \bar \Q_\perp(t) \right]$, where $\bar \Q(t)\in\mathds R^{n\times j^\star}$.
	
	System~\eqref{eq:sys:ltvwithoutput} is detectable, if the pair $(\B_1(t),\C(t)\bar \Q(t))$ with $\B_1$ as in~\eqref{eq:b1} and $k=j^\star$ is uniformly completely observable.
	In particular, there exists a bounded positive definite solution to the differential Riccati equation
	\begin{equation}\label{eq:Riccati_LTV}
		\dot \P\changed{_1} = \B_1(t)\P\changed{_1} + \P\changed{_1}\B_1(t) - \P\changed{_1}\bar\Q^\transp(t) \C^\transp(t)\C(t)\bar\Q(t) + \G\changed{_1}(t),
	\end{equation}
	with the $j^\star \times j^\star$ matrix $\P\changed{_1}(t_0)\succ 0$ and positive constants $g_1$, $g_2$ such that $g_1\I_{j^\star}\preceq \G\changed{_1}(t) \preceq g_2\I_{j^\star}$.
	Moreover, the error system~\eqref{eq:errordynamics} with
	\begin{equation}\label{eq:redgain}
		\L(t)=\bar\Q(t)\P\changed{_1}(t)\bar\Q\changed{_1}^\transp(t)\C^\transp(t)
	\end{equation}
	is uniformly exponentially stable.		
\end{theorem}
\begin{proof}
	Based on the Lyapunov transformation $\e(t) = \Q(t)\e_z(t)$, the error system~\eqref{eq:errordynamics} can be transformed to the block upper triangular form
	\begin{equation}
		\begin{bmatrix} \dot \e_{z,1}\\ \dot \e_{z,2}\end{bmatrix} = \begin{bmatrix}
			\B_1(t) & \B_{12}(t) \\
			\bm 0 & \B_2(t) 
		\end{bmatrix}	\begin{bmatrix} \e_{z,1}\\ \e_{z,2}\end{bmatrix} - \begin{bmatrix}
			\L_1(t) \\ \bm 0
		\end{bmatrix} \C(t) \begin{bmatrix}
			\bar \Q(t) & \bar \Q_\perp(t)
		\end{bmatrix}\begin{bmatrix} \e_{z,1}\\ \e_{z,2}\end{bmatrix}
	\end{equation}
	with $\L_1(t) = \P\changed{_1}(t)\bar\Q^\transp(t)\C^\transp(t)$.
	The subsystem $\dot\e_{z,2}=\B_2(t)\e_{z,2}$ is uniformly exponentially stable, because all upper Bohl exponents are negative.
	Hence, $\e_{z,2}(t)$ can be seen as a uniformly exponentially vanishing perturbation in the system
	\begin{equation}
		\dot \e_{z,1} = \left[\B_1(t) -\L_1\C(t)\bar \Q(t)\right] \e_{z,1} + \left[\B_{12}(t)-\L_1\C(t)\bar \Q_\perp(t)\right]\e_{z,2}
	\end{equation}
	Uniform complete observability of the pair $(\B_1(t),\C(t)\bar \Q(t))$ implies that there exists a uniformly bounded positive definite solution to~\eqref{eq:Riccati_LTV} and also implies that $\dot \e_{z,1} = \left[\B_1(t) -\L_1\C(t)\bar \Q(t)\right] \e_{z,1}$ is uniformly exponentially stable.
	Uniform exponential stability of the overall error system and hence detectability then follows from~\cite[Theorem 2]{zhou2016onasymptotic}.
\end{proof}
\begin{remark}
	For the practical implementation of the observer~\eqref{eq:observer} with the gain~\eqref{eq:redgain}, it suffices to solve the differential equation~\eqref{eq:redQR2} for some $k\geq j^\star$ instead of the (full order) differential equation~\eqref{eq:fullqr2}.
	This could drastically reduce the complexity, if $j^\star$ is small compared to the system order $n$.
	The parameter $k$ determines the number of upper Bohl exponents, which are modified by the observer feedback gain. 
	If system~\eqref{eq:sys:ltvwithoutput} is uniformly completely observable, $k$ can be regarded as a tuning parameter and allows a trade-off between the computational complexity and the convergence speed of the error system~\cite{tranninger2020detectability}.
	In particular, $k$ determines how many spectral intervals of the system are modified in the observer error dynamics. 
\end{remark}
In general, the detectability condition in~\Cref{thm:detectability:ltv} is not necessary, even if $\X_0$ is an ordered normal Lyapunov basis, as demonstrated by an example provided in~\cite{tranninger2021beobachterentwurf} \changed{where the upper Bohl exponents are not ordered.}
In the numerical simulations, it turned out that a certain ordering of the upper Bohl exponents similar to the Lyapunov exponents can be expected as shown in~\Cref{sec:example}. 
Moreover, if the system has an exponential dichotomy and $\X_0$ is an ordered normal Lyapunov basis, the detectability condition in~\Cref{thm:detectability:ltv} turns out to be necessary and sufficient~\cite{tranninger2020uniform}.

\changed{
	It is now shown, that the reduced order Riccati equation~\eqref{eq:Riccati_LTV} can be obtained from a projection of the full order Riccati equation with properly chosen $\P(t_0)$ and $\G(t)$.
	To that end, it is assumed that system~\eqref{eq:sys:LTV} is transformed to an upper triangular form with $\x = \Q(t)\z$ and partitioned according to $\Q(t) = \left[\bar \Q(t)\;\bar\Q_\perp(t) \right]$. 
	The full-order Riccati equation is then given by
	\begin{align}\label{eq:Riccati_LTV_full}
		&\begin{bmatrix}
			\dot \P_1 & \dot \P_{12} \\
			\dot \P_{12}^\transp &  \dot \P_2 
		\end{bmatrix} = \begin{bmatrix}
			\B_1 & \B_{12} \\
			\bf 0 & \B_2 
		\end{bmatrix}	\underbrace{\begin{bmatrix}
				\P_1 &  \P_{12} \\
				\P_{12}^\transp & \P_2 
		\end{bmatrix}}_{\P(t)} + \begin{bmatrix}
			\P_1 &  \P_{12} \\
			\P_{12}^\transp & \P_2 
		\end{bmatrix} \begin{bmatrix}
			\B^\transp_1 & \bf 0 \\
			\B^\transp_{12} & \B^\transp_2
		\end{bmatrix} \\
		&- \begin{bmatrix}
			\P_1 &  \P_{12} \\
			\P_{12}^\transp & \P_2 
		\end{bmatrix} \begin{bmatrix}
			\bar\Q^\transp \\
			\bar\Q_\perp^\transp 
		\end{bmatrix} \C^\transp \C \begin{bmatrix}
			\bar \Q& \bar \Q_\perp
		\end{bmatrix} \begin{bmatrix}
			\P_1 &  \P_{12} \\
			\P_{12}^\transp & \P_2 
		\end{bmatrix} + \begin{bmatrix}
			\G_1 & \G_{12} \\
			\G_{12}^\transp & \G_2
		\end{bmatrix},     \nonumber
	\end{align}
	where the time indices are omitted for the sake of readability.
	\begin{lemma}[Reduced Riccati equation]\label{le:Riccatireduced}
		Let $\P_1(t)$ be the solution of the differential Riccati equation~\eqref{eq:Riccati_LTV} and
		\begin{equation}\label{eq:Riccati_IC}
			\P(t_0) = \begin{bmatrix}
				\P_1(t_0) & \bf 0\\
				\bf 0 &\bf 0
			\end{bmatrix},\; \G_{12} = {\bf 0} \;\text{and}\;\; {\G}_2 = \bf{0}.
		\end{equation}
		Then, the solution of~\eqref{eq:Riccati_LTV_full} is given by
		\begin{equation}
			\P(t)	=\begin{bmatrix}
				\P_1(t) & \bm 0 \\
				\bm 0 & \bm 0
			\end{bmatrix} \quad \text{for all } t\geq t_0.
		\end{equation}
	\end{lemma}
	\begin{proof}
		The statement follows by substituting $\P(t)$ into~\eqref{eq:Riccati_LTV_full} and using~\eqref{eq:Riccati_LTV} and~\eqref{eq:Riccati_IC}.
	\end{proof}
} %

\section{Nonlinear Observer Design}\label{sec:NLobsv}
This section extends the observer design approach presented in~\Cref{sec:EDObserverdesign} to nonlinear systems.
As the Riccati differential equation is solved on a reduced state-space only, the ideas for the convergence proof of the extended Kalman-Bucy filter cannot be applied in the present setting.
However, the local stability result utilizes the Lyapunov function from~\Cref{thm:LyapLTV} and hence, under some assumptions stated subsequently, local uniform exponential stability of the observer error dynamics can be guaranteed. 

In the following, the nonlinear system~\eqref{eq:NL}, i.e.,
\begin{subequations}\label{eq:sysNLall}
	\begin{align}
		\dot \x &= \f(\x,\u), \,t\in\changed{\mathds R_{\geq 0}},\;\x(t_0)=\x_0,  \label{eq:sys_nonlinear}\\
		\y &= \C(t) \x\label{eq:sys_nonlinear:output}
	\end{align}	
\end{subequations}
is considered with the output equation~\eqref{eq:sys_nonlinear:output}.
Again, it is assumed that the output matrix function $\C:\changed{\mathds R_{\geq 0}}\rightarrow \mathds R^{p\times n}$ is bounded.
It should be remarked that the case of a linear measurement equation is used to demonstrate the key idea. 
An extension to a nonlinear output map is possible by a straightforward application of the concepts presented in~\cite{KONRAD1998}.

\subsection{Extended Kalman-Bucy Filter as a Deterministic Observer}\label{sec:EKBF}

The idea of the deterministic interpretation of the extended Kalman-Bucy filter is to obtain the time-varying matrices $\A(t)$ and $\C(t)$ by a linearization along the estimated trajectory~\cite{KONRAD1998}. 
The algorithm can be summarized as
\begin{subequations}\label{eq:EKBF}
	\begin{align}
		\changed{\dot{\hat \x}}  &= \f(\hat \x, \u)  + \P(t)\C^\transp(t)\left[\y - \C(t)\hat \x\right],\quad\hat \x(0) = \hat \x_0,\\
		\dot \P &= \P\A^\transp(t) + \A(t) \P - \P \C^\transp(t)\C(t) \P + \G(t), \quad \P(t)\in\mathds R^{n\times n}\label{eq:RiccatiEKBF}\\
		\A(t) &= \left.\pderiv{\phantom.}{\x} \f(\x(t),\u(t))\right|_{(\hat \x(t),\u(t))}.%
	\end{align}
\end{subequations}
The initial condition $\P_0$ is chosen as a positive definite matrix and the matrix $\G(t)$ is a positive definite tuning parameter.

Local stability results for the dynamics of the estimation error $\e(t) =\x(t)-\hat \x(t)$ are presented in~\cite{KONRAD1998,krener2003theconvergence,Bonnabel2015}.
A key assumption of these stability proofs is the existence of positive constants $p_1,\,p_2$ such that $p_1\I_n\preceq \P(t) \preceq p_2 \I_n$ holds for~\eqref{eq:RiccatiEKBF}.
This assumption can be guaranteed by uniform complete observability of the observer trajectory, which, however, might not be a trajectory of the original system in general.
This assumption cannot be dropped easily as extensively discussed in~\cite{Bonnabel2015,Karvonen2018}.

\subsection{Extended Subspace Observer}\label{sec:ESO}
Based on the ideas proposed in~\Cref{thm:detectability:ltv} and in analogy to the extended Kalman-Bucy filter, an observer for a class of nonlinear systems is presented in the following.
The proposed algorithm is summarized as
\begin{subequations}\label{eq:ESOnonlinear}
	\begin{align}
		\dot{\hat\x} &= \f(\hat\x,\u) + \L(t)\left[\y-\C(t)\hat\x\right], \\
		\dot{\P}\changed{_1} &= \B_{1}(t) \P\changed{_1} + \P\changed{_1} \B_{1}^\transp(t) - \P\changed{_1} \bar \C^\transp(t) \bar \C(t) \P + \G\changed{_1}(t),\;\quad \P\changed{_1}(t)\in\mathds{R}^{k\times k}\label{eq:RiccatiNonlinear}\\
		\dot{\bar \Q} &= \left[\I-\bar \Q \bar \Q^\transp\right]\A(t)\bar \Q +\bar \Q \S_{1}(t),\quad \bar \Q(t)\in\mathds{R}^{n\times k},
	\end{align}
	with 
	\begin{align}
		\L(t) &=\bar \Q(t) \P\changed{_1}(t)\bar \C^\transp(t),\label{eq:NLfeedback} %
		\quad\bar \C(t) = \C(t)\Q(t), \\
		\A(t) &= \left.\pderiv{\phantom.}{\x} \f(\x(t),\u(t))\right|_{(\hat \x(t),\u(t))}\label{eq:NLA(t)}, %
		\quad \B_{1}(t) = \bar \Q^\transp(t) \A(t) \bar \Q(t)- \S_{1}(t).		
	\end{align}
	The elements $s_{ij}$ of the skew symmetric matrix $\S_{1}$ are given by $s_{ij}(t)=\bar \q_i^\transp(t) \A(t) \bar \q_j(t)$ for $i>j$.
\end{subequations}
The number of columns in $\bar \Q(t)$ has to be chosen such that $j^\star \leq k \leq n$ with $j^\star$ as given in~\Cref{thm:detectability:ltv} and $\A(t)$ as in~\eqref{eq:NLA(t)}.
In the following, a convergence result for this observer is provided.

\subsection{Convergence Analysis}

The estimation error $\e(t)= \x(t)-\hat\x(t)$ is governed by
\begin{equation}
	\begin{aligned}\label{eq:nonlinerror1}
		\dot \e = \f(\x,\u)-\f(\hat\x,\u) - \L(t)\left[\y-\C(t)\hat\x\right].
	\end{aligned}
\end{equation}
By substituting $\x$ with $\x =\hat\x+\e$ in~\eqref{eq:nonlinerror1} and performing a Taylor series expansion of $\f(\cdot,\cdot)$ around $\e=\vc 0$ one obtains
$\f(\hat\x + \e, \u) = \f(\hat\x,\u) +\A(t)\e + \bm\eta(\e,\hat\x,\u),$ where $\A(t) = \left.\pderiv{\phantom.}{\x} \f(\x,\u)\right|_{(\hat \x(t),\u(t))}$
and $\bm \eta$ is the remainder of the Taylor series truncated after the linear term.
Hence, the estimation error dynamics can be stated as
\begin{equation}\label{eq:errorperturbed}
	\dot \e = \left[\A(t)-\L(t)\C(t)\right]\e + \bm\eta(\e,\hat\x,\u).
\end{equation}

The following additional assumptions are standard assumptions in the convergence analysis of the extended Kalman-Bucy filter, see, e.g.~\cite{KONRAD1998}.
\begin{enumerate}[label=(a\arabic*)]
	\item\label{as:NL1} The matrix function $\A:\changed{\mathds R_{\geq 0}} \rightarrow\mathds R^{n\times n}$ is bounded
	\item\label{as:NL2} For $\P\changed{_1}(t)$ in~\eqref{eq:RiccatiNonlinear}, there exist positive constants $p_1$ and $p_2$ such that $p_1\I_k \preceq \P\changed{_1}(t) \preceq p_2\I_k$ holds for all $t\in\changed{\mathds R_{\geq 0}}$.
	\item\label{as:NL3} There exist positive constants $\epsilon$ and $\kappa$ such that
	$\|\bm \eta(\e,\x,\u)\|  \leq \kappa \|\e(t)\|^2$
	holds for all $t\in \changed{\mathds R_{\geq 0}}$, $\x,\e\in\mathds R^n$ and $\u\in\mathds R^m$ with $\|\e(t)\|\leq \epsilon$.
\end{enumerate}
\begin{remark}\label{rem:assumptions}
	\changed{
		Assumption~\ref{as:NL1} is e.g. fulfilled if $\f$ is globally Lipschitz continuous in $\x$.
		The lower bound on $\P_1(t)$ in~\ref{as:NL2} is fulfilled for a properly chosen $\G_1(t)$ as shown in the proof of~\Cref{thm:ESOuniformlyobservable}.
		For a class of uniformly observable systems, the upper bound on $\P_1(t)$ moreover holds independently of the specific observer trajectory. 
		This result is stated in detail in~\Cref{sec:uniformlyobservable}.}
	Assumption~\ref{as:NL3} is fulfilled, e.g., if $\f$ is at least two times continuously differentiable and the corresponding Hessian matrix of each component of $\f$ is bounded, see~\cite{KONRAD1998,frank2018adetectability}. %
	Let the components of $\f(\x,\u)=\begin{bmatrix}f_1(\x,\u) & \cdots & f_n(\x,\u) \end{bmatrix}^\transp$ be denoted by $f_i$, $i=1,\,\ldots,\,n$.
	Then, $\kappa$ is given by
	\begin{equation}
		\kappa = \frac{1}{2}\max_{i=1}^n\sup_{\substack{\x\in\mathds R^n\\\u\in\mathds R^m}} \| \H_{f_i}(\x,\u) \| 
	\end{equation}
	with $\H_{f_i}(\x,\u)$ as the Hessian matrix of $f_i$.
\end{remark}

To show local stability of the estimation error dynamics, the standard approach  typically used in the stability analysis for the extended Kalman-Bucy filter, see, e.g.~\cite{KONRAD1998}, cannot be applied, because the Riccati equation is solved only on a reduced-order subspace.
However, a local convergence result is obtained by utilizing the Lyapunov function for linear time-varying systems from~\Cref{thm:LyapLTV}.

\begin{theorem}[extended subspace observer]
	Let a system be given by~\eqref{eq:sysNLall} and the observer for this system by~\eqref{eq:ESOnonlinear}.
	Moreover, let the assumptions (a1) to (a3) hold.
	Then, the estimation error dynamics~\eqref{eq:errorperturbed} resulting from this observer is locally uniformly exponentially stable.
\end{theorem}
\begin{proof}
	\changed{According to~\Cref{thm:detectability:ltv} and due to assumptions (a1) and (a2), it follows that} the linear time-varying system
	$\dot \e = \left[\A(t) - \L(t)\C(t)\right]\e$
	is uniformly exponentially stable with the feedback gain $\L(t)$ as in~\eqref{eq:NLfeedback}.
	The state transition matrix of this (unperturbed) system can then be bounded by
	\begin{equation}\label{eq:boundPhinonlinear}
		\|\mtPhi(t_1,t_0)\|\leq Ke^{-\gamma (t_1-t_0)}
	\end{equation}
	for some positive constants $K\geq 1$ and $\gamma>0$. 
	Moreover, according to~\Cref{thm:LyapLTV}, there exists a Lyapunov function $V(t,\e) = \e^\transp \PL(t) \e$ with positive constants $\bar p_{1}$ and $\bar p_{2}$ and a positive definite $n\times n$ matrix $\PL(t)$ such that $\bar p_{1}\I_n \preceq \PL(t) \preceq \bar p_{2}\I_n$.
	The matrix $\PL(t)$ is the unique positive definite solution of
	\begin{equation}
		\dot \P_\mathrm{L}+ \A^\transp_\mathrm{e}(t)\PL + \PL\A_\mathrm{e}(t) + \Q_\mathrm{L}(t) = \vc 0
	\end{equation}
	with $\A_\mathrm{e}(t) = \A(t)-\L(t)\C(t)$ and $\Q_\mathrm{L}(t)$ as any positive definite matrix bounded by positive constants $q_1$ and $q_2$ such that $q_{1}\I_n \preceq \Q_\mathrm{L}(t) \preceq q_2\I_n$.
	
	The function $V(t,\e)=\e^\transp \PL(t) \e(t)$ is now used as a Lyapunov function candidate for the perturbed error system~\eqref{eq:errorperturbed}.
	\changed{Let $\epsilon$ be given as in assumption (a3).}
	\changed{Then, for $\|\e\|\leq \epsilon$}, the time derivative along the trajectory can be obtained according to
	\begin{equation}
		\begin{aligned}
			\dot V(t,\e) &= \dot \e^\transp \PL \e + \e^\transp \dot \P_\mathrm{L} \e +\e^\transp \P_L \dot \e= -\e^\transp \Q_\mathrm{L} \e + 2\bm\eta(\e,\hat\x,\u)^\transp \PL \e \\
			&\leq -q_1\|\e\|^2+ 2 \bar p_{2}\|\bm\eta\| \|\e(t)\| \leq -q_1\|\e\|^2+2 \bar p_{2} \kappa \|\e\|^3
			\leq (-q_1+2\bar p_2 \kappa \|\e\|)\|\e\|^2 
		\end{aligned}
	\end{equation}
	For $\|\e\|\leq \min\left(\frac{q_1}{4\bar p_2\kappa},\epsilon\right)$, it holds that
	\begin{equation}\label{eq:boundderiv}
		\dot V(t,\e) \leq -\frac{1}{2}q_1\|\e\|^2.
	\end{equation} 
	According to~\cite[Theorem 4.10]{khalil2002nonlinear}, the error dynamics is locally uniformly exponentially stable because of~\eqref{eq:boundderiv} and $\bar p_{1} \|\e\|^2 \leq V(t,\e) \leq \bar p_{2}\|\e\|^2$.
\end{proof}

To further investigate the norm constraint on the initial error, it is now assumed for simplicity that $\epsilon \geq \frac{q_1}{4\bar p_2 \kappa}$ and that $\Q_\mathrm{L}(t) = q\I_n$, i.e., $q_1=q_2=q$.
Then, item~\ref{item:p2bound} of~\Cref{thm:LyapLTV}, i.e.,
$\bar p_2 \leq \frac{K^2 q}{2\gamma}$
with $K$ and $\gamma$ as in~\eqref{eq:boundPhinonlinear} can be used to simplify the constraint on the initial error according to
\begin{equation}
	\|\e(t_0)\| \leq \frac{q_1}{4\bar p_2\kappa} \leq \frac{\gamma}{2K^2\kappa}
\end{equation}
Here, $\gamma$ is the exponential decay rate.
For the proposed observer, $\gamma$ is bounded by the $k$-th upper Bohl exponent $\Bupper_k$ such that $\gamma<\beta_k$, because the $(k+1)$-th exponent is not modified by the observer gain.
This suggests that the region of convergence can be increased by taking more exponents into account, i.e., by increasing the number of columns in $\bar \Q$, if the system is not merely detectable but has stronger observability properties. 
This effect is also demonstrated in the numerical simulation examples.

\changed{\begin{remark}
		In general, the number of non-negative upper Bohl exponents $j^\star$ depends on the specific observer trajectory. 
		The Bohl exponents, however, are robust with respect to bounded vanishing perturbations, see~\cite[Theorem 5.2 and Remark 5.5]{daleckii2002stability}.
		Because the trajectories of the nonlinear system~\eqref{eq:sysNLall} are also trajectories of the observer (for zero initial error), it hence suffices to obtain $j^\star$ from simulation studies of the original system. 
		If the initial estimation error is small, the observer trajectory can then be regarded as perturbed trajectory of the true system trajectory with the same $j^\star$.
		The Bohl exponents of the linearizations along the trajectories in general depend on the specific trajectory.
		For the special class of ergodic measure preserving systems, Oseledet's multiplicative ergodic theorem guarantees that the Lyapunov exponents are independent of the specific trajectory, see~\cite[p. 8-9]{arnold1986lyapunov} or~\cite[Theorem 2.1]{johnson1987ergodic}.
		To the authors' knowledge, a similar result regarding the Bohl exponents or the exponential dichotomy spectrum is still missing.
		A step towards this direction is provided in~\cite{johnson1987ergodic}, where the relationship between Lyapunov exponents and the so-called dynamical spectrum is investigated.
		The dynamical spectrum is a generalization of~\Cref{def:dichotomyspectrum} to nonlinear systems and is independent of the specific linearization.
		It is hard to compute in practice, however.
	\end{remark}
}

\changed{
	\subsection{A Class of Uniformly Observable Systems}\label{sec:uniformlyobservable}
	For special classes of nonlinear systems, the assumptions~\ref{as:NL1} to~\ref{as:NL3} are guaranteed to hold and verifiable conditions depending only on the system properties can be stated. 
	In particular, this section considers nonlinear systems which are diffeomorphic to 
	\begin{subequations}\label{eq:sys:NLobsv}
		\begin{align}
			\dot \x  &= \bar\A \x + \bar \f(\x,\u)\\
			\y &= \bar\C \x
		\end{align}
		where the state $\x\in\mathds R^n$ is partitioned according to 
		\begin{equation}
			\x = \begin{bmatrix}\x^{(1)}\\ \x^{(2)}\\ \vdots \\ \x^{(p)}\end{bmatrix}\quad 
			\text{with } \x^{(k)} = \begin{bmatrix} x_{k,1}\\
				x_{k,2}\\
				\vdots \\
				x_{k,l_k}\end{bmatrix} \in\mathds R^{l_k},\, k=1,\ldots, p,\; \text{ and } \;\sum_k l_k = n.
		\end{equation}
		The rest of the system is characterized by
		\begin{align}
			\bar \A &= \diag{\bar \A_1,\, \bar \A_2,\,\ldots,\, \bar \A_p},\; \bar\A_k = \begin{bmatrix}
				0 & 1 & 0  &\cdots & 0 \\
				0 & 0 & 1 & \cdots & 0\\
				& & & \ddots & 0 \\
				0 &0&0& \cdots &0 
			\end{bmatrix}_{(l_k\times l_k)} \\
			\bar \C &= \diag{\bar \C_1,\, \bar \C_2,\, \ldots,\, \bar \C_p },\; \bar \C_k = \begin{bmatrix}
				1 & 0 & \cdots & 0 
			\end{bmatrix}_{(1\times l_k)}.
		\end{align}
		The nonlinear function $\bar \f$ is structured according to 
		\begin{equation}
			\bar \f(\x,\u) = \begin{bmatrix}
				\bar \f^{(1)}(\x,\u) \\ \bar \f^{(2)}(\x,\u) \\ \vdots \\ \bar \f^{(p)}(\x,\u)
			\end{bmatrix}
			\text{ with } \bar \f^{(k)}(\x,\u) = \begin{bmatrix}
				f_{k,1}(\x,\u)\\
				f_{k,2}(\x,\u)\\
				\vdots \\
				f_{k,l_k}(\x,\u)\\
			\end{bmatrix}
		\end{equation}
		and
		$\frac{\partial f_{i,j}}{\partial x_{v,w}} = 0$ for all $i,v \leq p$ and all $w>j$. 
	\end{subequations}
	Moreover, $\bar \f$ is assumed to be globally Lipschitz continuous in $\x$ and at least two times continuous differentiable, where the corresponding Hessian matrix of each component of $\bar \f$ is bounded.
	
	It is well known that for general nonlinear systems, the input may destroy the system's observability. 
	Systems diffeomorphic to~\eqref{eq:sys:NLobsv} are called uniformly observable for any input~\cite{krener2003theconvergence}, because observability does not depend on the particular input.
	Systems in the form~\eqref{eq:sys:NLobsv} play an important role in the observer design for nonlinear systems, see, e.g.~\cite{gauthier1992asimple,hammouri2003nonlinear,farza2004observer,hammouri2010highgain,gauthier2001deterministic,bernard2022observer}.
	A lot of work has been devoted to determining existence conditions for transformations to the form~\eqref{eq:sys:NLobsv} or more general upper triangular forms, see~\cite[Sec. 6]{bernard2022observer} and the references therein for a good overview on existing procedures.
	It is worth to mention that system~\eqref{eq:sys:NLobsv} is a generalization of the phase variable form~\cite{bestle1983canonical}.
	
	For system~\eqref{eq:sys:NLobsv}, Krener~\cite{krener2003theconvergence} shows that the extended Kalman-Bucy filter~\eqref{eq:EKBF} is a local exponential observer.
	In particular, the proof of~\cite[Theorem 1.1.1]{krener2003theconvergence} establishes the boundedness of $\P(t)$ in~\eqref{eq:RiccatiEKBF} independently of the observer trajectory. 
	This important result is summarized in 
	\begin{lemma}[boundedness of $\P(t)$]\label{le:Riccatibound}
		Assume that the nonlinear system is in the form~\eqref{eq:sys:NLobsv}.
		Moreover $\bar \f$ is globally Lipschitz continuous and at least twice continuously differentiable and the corresponding Hessian matrix of each component of $\bar \f$ is bounded.
		Then, for all $\P(t_0)\succeq \bm 0$ and all bounded $\G(t)\succeq \bm 0$, there exists a constant $p_2>0$ such that for the solution $\P(t)$ of the corresponding Riccati equation~\eqref{eq:RiccatiEKBF} it holds that $\P(t)\leq p_2\I_n$ for all $t\geq t_0$.
	\end{lemma}
	\begin{proof}
		The proof for the bound with $\P(t_0)\succ 0$ is given in the proof of~\cite[Theorem 1.1.1]{krener2003theconvergence}.
		Remarkably, this bound is independent of the observer trajectory and hence it holds for any trajectory.
		The bound is established via the design of a suboptimal high-gain observer, which gives an upper bound on the cost of the optimal filtering problem and hence an upper bound on the solution of the Riccati equation~\eqref{eq:RiccatiEKBF}.
		The boundedness for all $\P(t_0)\succeq \bm 0$ then follows from the order preserving property of the Riccati differential equation, see~\cite[Theorem 4.1.4 and Corollary 4.1.5]{aboukandil2012matrix}.
	\end{proof}
	The previous result together with the specific system structure guarantees that the assumptions~\ref{as:NL1} to~\ref{as:NL3} are fulfilled.
	Hence, the existence of the extended subspace observer proposed in Section~\ref{sec:ESO} is guaranteed for this system class.
	This is summarized in the following
	\begin{theorem}[existence of extended subspace observer]\label{thm:ESOuniformlyobservable}
		For systems of the form~\eqref{eq:sys:NLobsv}, where $\bar\f$ is globally Lipschitz continuous and at least twice continuously differentiable and the corresponding Hessian matrix of each component of $\bar \f$ is bounded, the assumptions~\ref{as:NL1}--\ref{as:NL3} are fulfilled for all $\P_1(0)\succ \bm 0$ and all $g_1\I\preceq \G_1(t)\preceq g_2 \I$ with arbitrary constants $0<g_1\leq g_2$.	
	\end{theorem}
	\begin{proof}
		Assumptions~\ref{as:NL1} and~\ref{as:NL3} are guaranteed by the global Lipschitz assumption and the boundedness assumption on the Hessian, see also~\Cref{rem:assumptions}.
		Hence it remains to show that~\ref{as:NL2} holds. 
		The upper bound on $\P_1(t)$ follows straightforwardly from~\Cref{le:Riccatireduced} and~\Cref{le:Riccatibound} for any positive definite and bounded $\G_1(t)$.
		The proof that $\P_1(t)$ is also bounded from below follows from continuity arguments.
		Under the assumption that $\P_1(t)$ becomes positive semidefinite for some $t=t_1 > t_0$, there exists a non-trivial vector $\v$ such that $\v^\transp \P_1(t_1) \v = 0$. 
		Multiplying the Riccati equation~\eqref{eq:RiccatiNonlinear} with $\v^\transp$ from the left and $\v$ from the right and evaluating the derivative at $t=t_1$ gives
		\begin{equation}
			\frac{\mathrm d}{\mathrm d t}\v^\transp \P_1 \v\big|_{t=t_1}=\v^\transp \G_1(t_1) \v \succ \bm 0.
		\end{equation}
		Due to continuity and since $\v^\transp \P_1(0)\v \succ > 0$, the solution of~\eqref{eq:RiccatiNonlinear} must remain positive definite and lower bounded according to $p_1\I\preceq \P_1(t)$ with some $p_1>0$.
	\end{proof}
	\Cref{thm:ESOuniformlyobservable} guarantees the existence of the extended subspace observer for systems in the form of~\eqref{eq:sys:NLobsv}.
	The global Lipschitz continuity of $\bar \f$ is also required in high-gain observer designs, see, e.g.,~\cite{farza2004observer}.
	This condition is restrictive since the Lipschitz condition is usually satisfied only locally. 
	If, however, the state $\x(t)$ of system~\eqref{eq:sys:NLobsv} lies in a bounded set, it is possible to relax this assumption using prolongation techniques, see~\cite[Remark 2.1]{farza2004observer} and the references therein.
}
\section{Numerical results}\label{sec:example}
The nonlinear Lorenz'96 (L'96) model introduced in~\cite{lorenz96predictability} is widely used as a benchmark of data assimilation algorithms~\cite{trevisan2011onthekalman,carrassi2008dataassimilation,frank2018adetectability}.
It is a system of nonlinear differential equations recursively defined by 
\begin{equation}\label{eq:lorenz96nonlinear}
	\dot{z}_i=(z_{i+1}-z_{i-2})z_{i-1}-z_i + F,\qquad i=1,\ldots,n,
\end{equation}
with the notational convention $z_{-1}= z_{n-1}$, $z_{0}= z_n$, and \mbox{$z_{n+1}= z_1$}. 
The state vector is $\z = [z_1\;\cdots \; z_n]^\transp\in\mathds{R}^n$.
For $F=8$, this model exhibits a chaotic behavior.
With this configuration, the model is used as a benchmark examples for data assimilation in~\cite{frank2018adetectability,palatella2013Lyapunov,Bocquet2017a}.

The following simulations are carried out for a model of order $n=18$. 
The output of the model is chosen as $\y(t)=\C_p \z(t)$ with a constant $p\times n$ matrix $\C_p$. 
The rows of $\C_p$ are chosen such that only one state is measured per row and the ``sensors'' are distributed with equal distance over all state variables.
The matrix is given by
\begin{equation}\label{eq:simoutput}
	\C_p = \begin{bmatrix} \b_1 & \b_{d+1} & \cdots & \b_{(p-1)d+1}  \end{bmatrix}^\transp,
\end{equation} 
where $\b_i$ is the i-th standard basis vector and $d=\left\lfloor \frac{n}{p}\right\rfloor$, where $\lfloor \cdot \rfloor$ denotes rounding to the nearest integer towards zero.

The initial condition is chosen as $z_{i,0}= \sin\left(\frac{i-1}{n}2\pi\right)$ at initial time $t_0=0$.
For evaluation purposes, also different initial conditions were chosen. 
The obtained results were quantitatively similar and in particular the approximated Lyapunov exponents and the upper and lower Bohl exponents were approximately the same. 

First, the spectral intervals of the system are investigated. 
All differential equations are solved using fixed step 4th order Runge-Kutta integration with a step-size of $T_s = 0.005$. 
For solving the differential equation \eqref{eq:redQR2} to obtain $\bar \Q$, a projected integrator~\cite{Dieci1994} is implemented, where the orthonormalization is carried out after each simulation time step.

The results for all $n=18$ spectral intervals are given in~\Cref{tab:spectralintervalsL96} in \Cref{sec:appendix_results}. 
The final time of this simulation is chosen as $T_f=1500$ in order to be able to choose the averaging window length $H$ for the approximation of the exponential dichotomy spectral intervals sufficiently large. 
For a window length of $H=300$, the $7^\text{th}$ approximated upper Bohl exponent is positive whereas for an increased window length $H=800$, this value is negative.
The spectral intervals $\Lambda_7$ to $\Lambda_{18}$ are thus negative and moreover all intervals are disjoint. 
The results show that the number of non-negative Lyapunov exponents $k^\star$ is equal to the number of non-negative upper Bohl exponents $j^\star$.

Now, the design of the extended subspace observer is carried out for $p=5$ as the dimension of the output vector $\y(t)$. 
The initial estimate is chosen as $\hat \x_0 = \x_0 + \bm \xi_0$ with $\bm \xi_0$ as a random perturbation. 
The components $\xi_{0,i}$ of $\bm \xi_0$ are chosen to be uniformly distributed on an interval $(-\delta,\delta)$ with $\delta>0$ as a simulation parameter.
For $k\leq 5<j^\star$, no convergence of the estimation error could be achieved independently of the size of the initial perturbation. 
This coincides with the observations from the approximated spectral intervals.
Hence, in a first step, $k=k^\star=j^\star=6$  is chosen in the observer design and an ensemble simulation with $N=50$ simulation runs is carried out.
The matrix $\G_{\changed{1}}(t)$ in the differential Riccati equation~\eqref{eq:RiccatiNonlinear} is chosen as a constant matrix $\G_{\changed{1}} = 10\I_k$ in the following simulations.

The (point-wise) minimum and maximum estimation errors together with the median and the 80\%-quantile are depicted in Fig.~\ref{fig:l96ensemblek6}. 
The latter quantity is a (point-wise) upper bound for 80\% of the estimation errors.
The expected minimum convergence rate of the observer error, i.e., the approximated 7-th upper Bohl exponent is also depicted.
This indicates that the estimation error converges exponentially with the rate $\beta_7^{800}\approx -0.0320$.
This convergence rate is very small and, as predicted by theory, the stability of the nonlinear error system is very sensitive to the size of the initial estimation error $\e_0$. 
This behavior can be seen in the simulations and hence the interval bound for the uniform distribution of the initial error is chosen as $\delta=10^{-4}$.

In order to increase the convergence speed and to decrease sensitivity with respect to the magnitude of the initial perturbation, the number of columns in $\bar \Q$ was increased to $k=7$.
The results of this ensemble simulation is depicted in Fig.~\ref{fig:l96ensemblek7}.
The resulting exponential convergence is now achieved at an approximate rate of $\beta_8^{800}\approx-0.2760$, because this is the first Bohl exponent which is not modified by the observer feedback gain.
The convergence is faster compared to $k=6$ and moreover the robustness with respect to the initial perturbation is improved. Hence, the size of the initial perturbation could be increased to $\delta=10^{-3}$. 
\changed{Towards the end of the simulation, the estimation error norm saturates around $10^{-13}$ due to numerical integration and floating point error. 
}

Taking into account one more spectral interval to be modified via the observer gain, the convergence speed can again be increased as depicted in Fig.~\ref{fig:l96ensemblek8} for $k=8$.
The size of the initial perturbation was again increased to $\delta=10^{-2}$ and a convergence of the estimation error was achieved for all simulation runs.  
\begin{figure}
	\centering
	\includegraphics[width=0.85\textwidth]{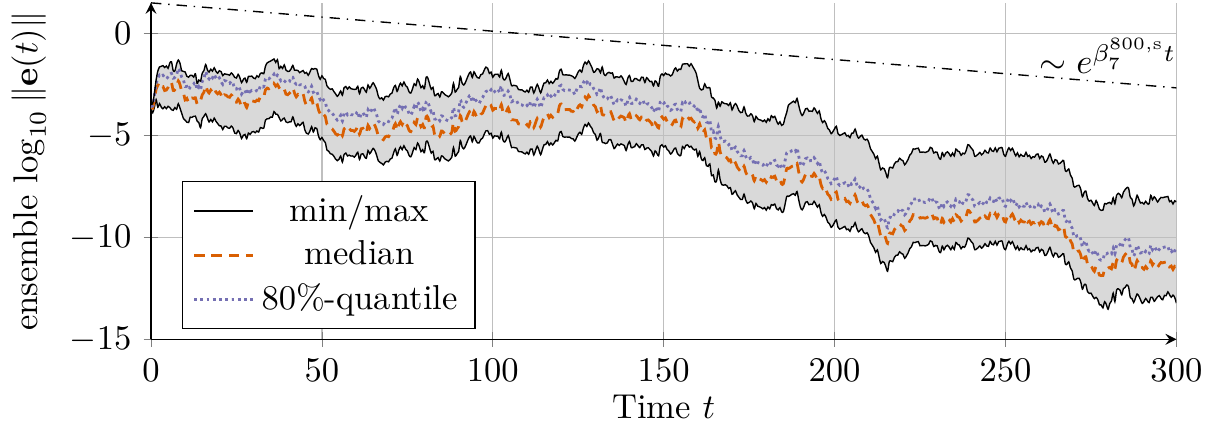}
	\caption{Ensemble estimation error for the L'96 model with $k=6$.}
	\label{fig:l96ensemblek6}
\end{figure}
\begin{figure}
	\centering
	\includegraphics[width=0.85\textwidth]{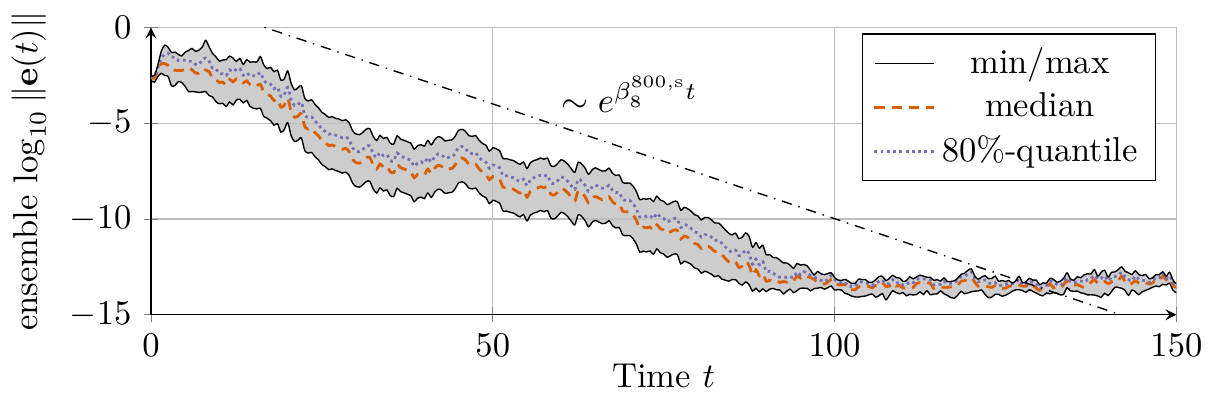}
	\caption{Ensemble estimation error for the L'96 model with $k=7$.}
	\label{fig:l96ensemblek7}
\end{figure}
\begin{figure}
	\centering
	\includegraphics[width=0.85\textwidth]{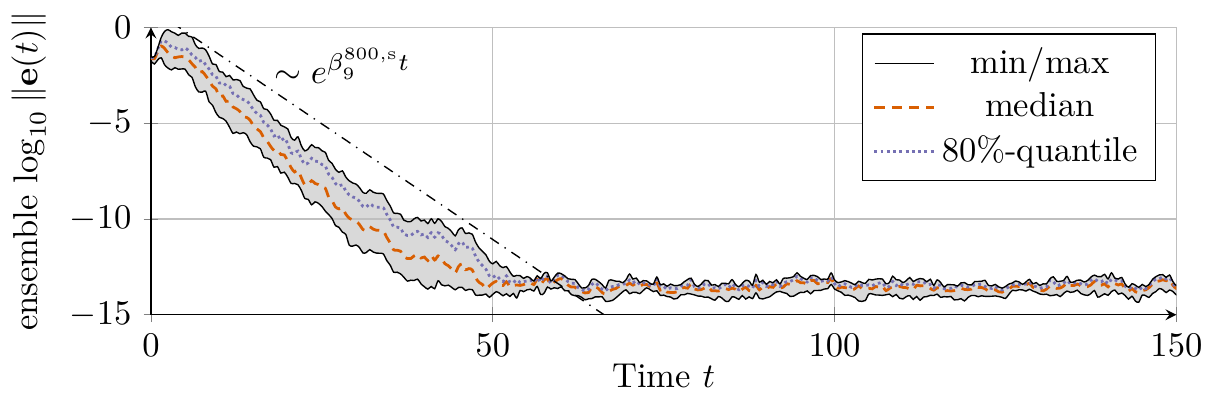}
	\caption{Ensemble estimation error for the L'96 model with $k=8$.}
	\label{fig:l96ensemblek8}
\end{figure}
In order to analyze the detectability properties of the Lorenz'96 model, the observer was implemented for $k=j^\star=6$ with the matrix $\G_{\changed{1}}(t)=\vc 0$ in~\eqref{eq:RiccatiNonlinear}.
Then, the obtained differential Riccati equation is a differential equation for the inverse of the constructibility Gramian on the considered subspace.
This inverse, i.e., $\P\changed{_1}(t_0)$, cannot be initialized correctly, because the constructibility Gramian is zero at the initial time. 
\changed{However, independently of the positive definite initialization, all solutions of $\P\changed{_1}(t)$ converge to each other and hence to the true inverse asymptotically, see~\cite[Sec. 3]{reddy2020asymptotic}.
}
This suggests that if $\P_{\changed{1}}(t)$ is upper bounded, the constructi\-bility Gramian is lower bounded and the linearization along the estimated trajectory is uniformly completely constructible on the considered subspace.
The largest eigenvalue of $\P\changed{_1}(t)$ for different measurement configurations is depicted in Fig.~\ref{fig:l96sigmapmax}. 
This result indicates that the considered trajectory is uniformly completely observable on the unstable subspace for any number of measurements. 
\begin{figure}[htbp]
	\centering
	\includegraphics[width=0.75\textwidth]{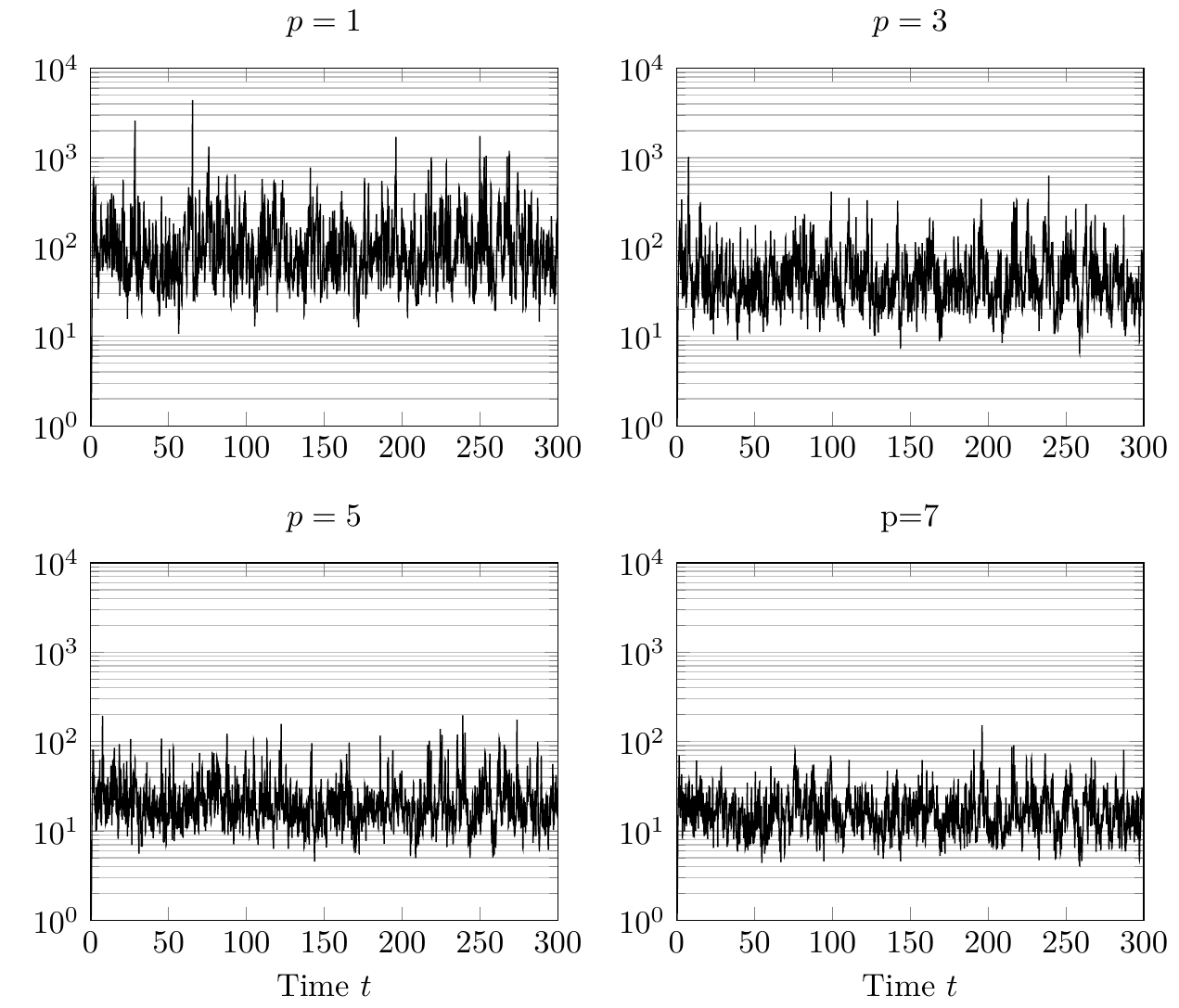}
	\caption{Maximum eigenvalue of $\P\changed{_1}(t)$ for the L'96 model with different measurement configurations.}
	\label{fig:l96sigmapmax}
\end{figure}
An additional insight can be obtained by the investigation of the smallest eigenvalue of $\P\changed{_1}(t)$. 
The number of measurements was now again chosen to be $p=5$.
If one considers only the unstable modes in the observer gain, i.e. $k=6$, the smallest eigenvalue is uniformly bounded from below. 
However, if one takes into account an additional mode in the observer gain by choosing $k=7$, the smallest eigenvalue tends to zero, see Fig.~\ref{fig:l96sigmapmin}. 
This indicates that the observer gain is ``losing strength'' in the already uniformly asymptotically stable directions. 
This is avoided by choosing $\G(t)$ as a positive definite matrix as proposed in the present observer design.
A similar effect is also discussed for a deterministic interpretation of the Kalman-Bucy filter in \changed{\cite[Remark 3.2]{reddy2020asymptotic} and is also well recognized in data assimilation on the unstable subspace~\cite{trevisan2011onthekalman,Bocquet2017}}.
\begin{figure}[htbp]
	\centering
	\includegraphics[width=0.75\textwidth]{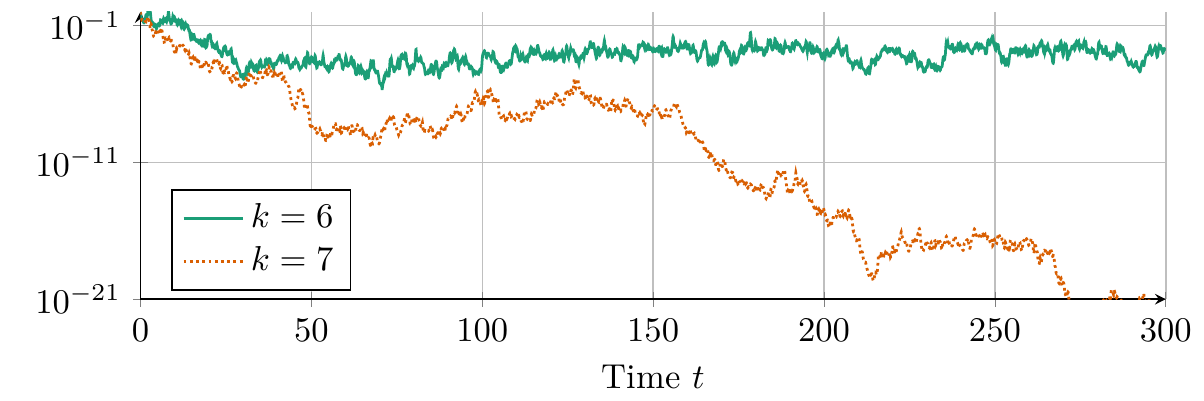}
	\caption{Minimum eigenvalue of $\P\changed{_1}(t)$ for the L'96 model with different values for $k$ and $\Q(t)=\vc 0$.}
	\label{fig:l96sigmapmin}
\end{figure}

\changed{To test the observer for various system trajectories, the upper Bohl exponents for random initial conditions on the unit sphere of the L'96 system were computed. 
	The maximum, minimum and average values of the 9 largest upper Bohl exponents for 500 simulation runs is given in~\Cref{tab:upperBohlL96}.
	For all computed trajectories, the $8^{\text{th}}$ Bohl exponent was negative and hence $k=7$ is a reasonable choice.
	The ensemble estimation error for $50$ observer simulation runs and an initial perturbation of $\delta=10^{-4}$ is depicted in~\Cref{fig:l96ensemblek7robustness}.
	The dashed blue line shows a convergence rate proportional to the mean over all $\beta_8^{900,\mathrm{s}}$, which suggests that this exponent governs the convergence behavior.
}
\begin{figure}
	\centering
	\includegraphics[width=0.85\textwidth]{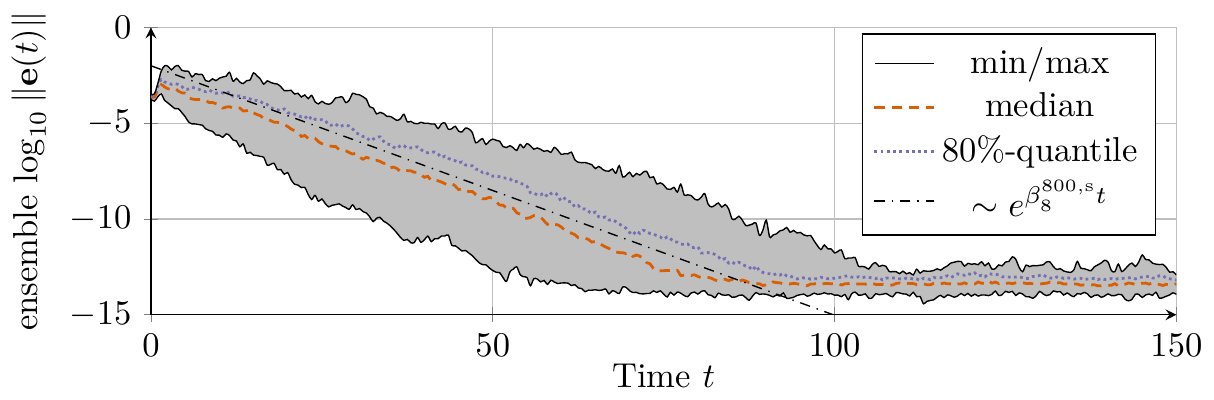}
	\caption{Ensemble estimation error for the L'96 model with $k=7$ and random initial conditions.}
	\label{fig:l96ensemblek7robustness}
\end{figure}

\section{Conclusions and Outlook}
\label{sec:conclusions}

This work proposes a new observer design strategy for linear time-varying and nonlinear systems.
Stability proofs for the estimation error dynamics show global exponential stability in the time-varying case and local exponential stability in the nonlinear case.
The feedback gain is computed only for the unstable modes of the estimation error dynamics, which allows to reduce the computational complexity of the observer compared to classical approaches like the extended Kalman-Bucy filter.
Simulation results moreover suggest that additionally including stable modes in the design may increase the robustness in the nonlinear case.
In future research, the assignment of the exponential dichotomy spectrum will be investigated. 
This may foster the development of efficient observer designs, which do not require the solution of a (computationally costly) differential Riccati equation.
\changed{Moreover, the conditions for the existence of the proposed observer will be investigated in more detail.
For this purpose, it is important to guarantee the upper bound of the observer Riccati equation independent of the observer trajectory. 
For systems which are fully observed on the unstable subspace, it may be possible to guarantee this upper bound by adapting the bound from~\cite{karvonen2018bounds}.}

\appendix

\section{Proof of~{\Cref{thm:approximation}}}
The proof for relation~\eqref{eq:bapprox1} can be found in the proof of~\cite[Theorem 8.4]{Dieci2003}. 
Relation~\eqref{eq:bapprox2} is proven for one upper Bohl exponents $\Bupper_j$. 
First, assume an arbitrary $\epsilon>0$. 
The scalar system $\dot \eta(t) = (b_{jj}(t)-\lambda)\eta(t)$ with $\lambda = \Bupper_j+\epsilon$ has an exponential dichotomy with $P=1$ and, equivalently, there exist reals $\alpha>0$ and $d\geq 0$ such that 
\begin{equation}\label{eq:integralsep}
	-\ln\| \mtPhi(t,t_0)\| =\integ{s}{t_0}{t}{\left(\Bupper_j+\epsilon-b_{jj}(s)\right)}\geq \alpha(t-t_0)-d
\end{equation}
holds for all $t_0\in\mathds J$ and all $t\geq t_0$. 
For a sufficiently large $H_0>0$ such that $\alpha-d/H_0>\alpha/2$, it follows that
\begin{equation}\label{eq:prf1}
	\frac{1}{H}\integ{s}{t}{t+H}{\left(\Bupper_j+\epsilon-b_{jj}(s)\right)} \geq \alpha-\frac{d}{H}>\frac{\alpha}{2}
\end{equation}
holds for all $H\geq H_0$.
Splitting the integral in~\eqref{eq:prf1} yields
\begin{equation}
	\frac{1}{H}\integ{s}{t}{t+H}{\left(\Bupper_j +\epsilon\right)} - \frac{1}{H}\integ{s}{t}{t+H}{b_{jj}(s)} =\Bupper_j+\epsilon-\Bnumup_j >\frac{\alpha}{2},
\end{equation}
and hence $\Bupper_j+\epsilon > \Bnumup_j$ for all $H\geq H_0$.
The proof for the remaining upper and lower Bohl exponents follows analogously.
\hfill \qed

\section{Numerical Results}\label{sec:appendix_results}

\begin{table}[htpb]
	\centering
	\caption{\changed{Nine largest} approximated spectral intervals of the Lorenz'96 model \changed{with initial condition \mbox{$z_{i,0}=\sin(\frac{i-1}{n}2\pi)$}}.}\label{tab:spectralintervalsL96}
	\resizebox{0.65\columnwidth}{!}{%
		\renewcommand*{\arraystretch}{1.2}
		\begin{tabular}{|c||c||c|c||c|c|}\hline
			& $\SigmaL$  & \multicolumn{2}{c||}{$\SigmaED$ for ${H=300}$}& \multicolumn{2}{c|}{$\SigmaED$ for ${H=800}$} \\\hline
			interval & $\Llower_i\approx \Lupper_i$ & $\beta_i^{300,\mathrm{i}}$ & $\beta_i^{300,\mathrm{s}}$ & $\beta_i^{800,\mathrm{i}}$ & $\beta_i^{800,\mathrm{s}}$ \\ \hline
			1&	\phantom{-}1.545&    \phantom{-}1.430&    \phantom{-}1.658&   \phantom{-}1.4870  &  \phantom{-}1.5920\\ \hline
			2&	\phantom{-}1.211 &   \phantom{-}1.139&    \phantom{-}1.274&   \phantom{-}1.1920 &   \phantom{-}1.2410\\ \hline
			3&	\phantom{-}0.878 &   \phantom{-}0.790&    \phantom{-}0.959&   \phantom{-}0.8630 &   \phantom{-}0.9190\\ \hline
			4&	\phantom{-}0.570&    \phantom{-}0.494&    \phantom{-}0.643&   \phantom{-}0.5360 &   \phantom{-}0.5930\\ \hline
			5&	\phantom{-}0.283 &   \phantom{-}0.231&   \phantom{-}0.333&    \phantom{-}0.2600  &  \phantom{-}0.2900\\ \hline
			6&	\phantom{-}0.003 &  -0.011&   \phantom{-}0.021&   -0.0050 &  \phantom{-}0.0080	\\ \hline
			7&	-0.045&  -0.096&    \phantom{-}0.015&    -0.0760 &  -0.0320\\ \hline
			8&	-0.296&  -0.365&   -0.228&   -0.3310  & -0.2760	\\ \hline
			9&	-0.579&  -0.639&   -0.526&   -0.5880  & -0.5560	\\ \hline
		\end{tabular}
	}
\end{table}

	\begin{table}	\caption{Nine largest upper Bohl exponents of Lorenz'96 model with random initial condition.}\label{tab:upperBohlL96}
		\centering
		\resizebox{0.8\columnwidth}{!}{%
			\renewcommand*{\arraystretch}{1.2}
			\begin{tabular}{|c||c|c|c|c|c|c|c|c|c|}\hline
				& $\beta_1^{800,\mathrm{s}}$ & $\beta_2^{800,\mathrm{s}}$ & $\beta_3^{800,\mathrm{s}}$ 	 & $\beta_4^{800,\mathrm{s}}$ 	 & $\beta_5^{800,\mathrm{s}}$ 	& $\beta_6^{800,\mathrm{s}}$ 	& $\beta_7^{800,\mathrm{s}}$ 	& $\beta_8^{800,\mathrm{s}}$ 	& $\beta_9^{800,\mathrm{s}}$ \\ \hline
				max &   1.669 &1.266 &	0.932 &	0.628&	0.344&	0.046&	0.015&	-0.246&	-0.518 	\\  \hline
				min &	1.478 &	1.112&	0.813&	0.511&	0.216&	-0.005&	-0.074&	-0.355&	-0.641 \\ \hline
				mean &   1.552& 1.185&    0.869&    0.569&    0.274&    0.004&    -0.022&    -0.307&    -0.576\\ \hline
			\end{tabular}
		}
	\end{table}

\section*{Acknowledgments}
This work was partially supported by the Graz University of Technology LEAD project ``Dependable Internet of Things in Adverse Environments''. The financial support by the Christian Doppler Research Association, the Austrian Federal Ministry for Digital and Economic Affairs and the National Foundation for Research, Technology and Development is gratefully acknowledged.
The authors would moreover like to thank Sergiy Zhuk, IBM Research, Dublin, Ireland, for the many fruitful discussions leading to this interesting research direction.

\bibliographystyle{siamplain}
\bibliography{references}

\end{document}